\documentclass{lmcs}
\pdfoutput=1
\usepackage[utf8]{inputenc}

\usepackage{lastpage}
\lmcsdoi{19}{4}{13}
\lmcsheading{}{\pageref{LastPage}}{}{}%
{Nov.~01,~2022}{Nov.~22,~2023}{}

\keywords{actor system, asynchronous process calculus, behavioural equivalence, barbed bisimulation, fail-stop, failure injection, grey failure, recovery}

\usepackage{xspace}
\usepackage{amssymb}
\usepackage{hyperref}
\usepackage[noabbrev,capitalize]{cleveref}
\usepackage{subcaption}
\usepackage{mathpartir}
\usepackage[final]{listings}
\lstdefinestyle{erl}{basicstyle=\small\ttfamily,
  keywordstyle=\color{purple}\textbf,
  commentstyle=\color{green}\textit,
  stringstyle=\color{red},
  language=erlang
}

\usepackage{macros}

\begin{document}

\title{A model of actors and grey failures}
\thanks{This work
    has been partially supported by EPSRC project EP/T014512/1 (STARDUST) and the BehAPI project funded by the EU H2020 RISE under the Marie Sklodowska-Curie action (No: 778233).}
\author[L.~Bocchi]{Laura Bocchi\lmcsorcid{0000-0002-7177-9395}}[a]
\author[J.~Lange]{Julien Lange\lmcsorcid{0000-0001-9697-1378}}[b]
\author[S.~Thompson]{Simon Thompson\lmcsorcid{0000-0002-2350-301X}}[a]
\author[L.~Voinea]{A. Laura Voinea\lmcsorcid{0000-0003-4482-205X}}[a,c]

\address{University of Kent, Canterbury, UK}
\email{l.bocchi@kent.ac.uk, s.j.thompson@kent.ac.uk}

\address{Royal Holloway, University of London, Egham, UK}
\email{julien.lange@rhul.ac.uk}

\address{University of Glasgow, Glasgow, UK}
\email{laura.voinea@glasgow.ac.uk}

\begin{abstract}
Existing models for the analysis of concurrent processes tend to focus on fail-stop failures, where processes are either working or permanently stopped, and their state (working/stopped) is known.
In fact, systems are often affected by grey failures: failures that are latent, possibly transient, and may affect the system in subtle ways that later lead to major issues, such as crashes, limited availability or overload. 

We introduce a model of actor-based systems with grey failures, based on two interlinked layers:
an actor model, given as an asynchronous process calculus with discrete
time, and a failure model that represents failure patterns that can be injected into the
system. Our failure model captures not only fail-stop node and link failures, but also grey failures, which might be partial or transient. 

We give a behavioural equivalence relation based on weak barbed bisimulation to compare systems on the basis of their ability to recover from failures, and on this basis we define some desirable properties of reliable systems. By doing so, we reduce the problem of checking reliability properties of systems to the problem of checking bisimulation.
\end{abstract}

\maketitle

\section{Introduction}

Many real-world computing systems are affected by non-negligible degrees of unpredictability, such as unexpected delays and failures, which are not straightforward to capture accurately.
Several works contribute towards a formal account of unpredictability, for example in the context of process calculi -- potentially including session types -- by extending calculi to model node failures~\cite{FournetGLMR96,ICALP-1997-RielyH}, link failures~\cite{DBLP:conf/forte/AdameitPN17}, and a combination of link and node failures~\cite{BergerH00}; these calculi also add a variety of program constructs to deal with failures including escapes~\cite{capecchi_giachino_yoshida_2016}, interrupts~\cite{10.1007/978-3-642-40787-1_8}, exceptions~\cite{DBLP:journals/pacmpl/FowlerLMD19}, and timeouts~\cite{DBLP:conf/fossacs/LaneveZ05,DBLP:conf/aplas/BergerY07,DBLP:conf/wsfm/LopezP11}.
Most existing models assume a fail-stop model of failure, where processes are either working or permanently stopped, and their state of being either working or stopped is known. In fact, systems are often affected by grey failures: failures that are latent, possibly transient, and may affect the system in subtle ways that later lead to major issues, such as crashes, limited availability and overload.  The symptoms of grey failure tend to be ambiguous. Several kinds of grey failure have been studied in the last decade such as transient failure (e.g., a component is down at periodic intervals), partial failure (only some sub-components are affected), or slowdown~\cite{slowdown}. In a distributed system, processes may have different perceptions as to the state of health of the system (aka \emph{differential observation})~\cite{achilles}. Grey failures tend to be behind many service incidents in cloud systems, and in these situations traditional fault tolerance mechanisms tend to be ineffective or counterproductive~\cite{achilles}. Diagnosis can be challenging and lengthy: for example, the work in~\cite{partial} estimates a median time for the diagnosis of partial failures to be 6 days and 5 hours. One of the main causes of late diagnosis is ambiguity of the symptoms and hence difficulty in correlating failures with their effects.

In this paper we make a first step towards a better understanding of the correlation between failures and symptoms via static formal analysis. We focus on the distributed actor model of Erlang~\cite{Armstrong13}, which is known for its effectiveness in handling failures and has been emulated in many other languages, e.g., the popular Akka framework for Scala~\cite{wyatt2013akka}.

We define a formal model of actor-based systems with grey failures, which we call \emph{`cursed systems'}. More precisely, we introduce two interlinked models: (1) a \emph{model of systems}, which are networks of distributed actors; (2) a \emph{model of (grey) failures} that allows us to characterise \emph{`curses'} as patterns of grey failures to inject in the system. This model of failures can represent node failures (with loss of messages in the node's mailbox), node slowdowns, link failures (with loss of the message in transit), and link slowdowns. The aforementioned instances of failure can be specified at the granularity of single nodes and links, to capture total and partial failures, and at the granularity of  (discrete) time instants, to capture permanent, transient, and periodic failures. For example, a failed node can be in a  failed state for a while before being restarted. The model of systems allows one to specify whether a node is restarted from the initial state (reset) or from a checkpoint. To capture the ambiguity of symptoms of grey failure we assume that actors have no knowledge of the state of health of other actors. However, actors can observe the presence (or absence) of messages in their own mailboxes and hence can infer the effects of failure from the communications that they have (not) received. In Erlang, a key mechanism for detecting and dealing with failure is the use of timeouts, which are one of the main ingredients of our system model.

Modelling failures as a separate layer allows us to compare systems recovery strategies with respect to specific failure patterns. This is a first step towards analysing the resilience of  systems to failures, and assessing the effect of failure on different parts of the system. We introduce a behavioural equivalence, based on weak barbed bisimulation, to compare systems affected by failures. We show that reliability properties of interest, namely resilience and recoverability, can be reduced to the problem of checking weak barbed bisimulation between  systems with failures.
Furthermore, we introduce a notion of augmentation, based on weak barbed bisimulation, to model and analyse the improvement of a system with respect to its recoverability against certain kinds of failure.

\paragraph{Synopsis} The paper is structured as follows. In Section~\ref{sec:overview}, we give an informal overview of the system model, and compare it with related work. Next we introduce the models of failure (Section~\ref{sec:failures}) and systems (Section~\ref{sec:calculus}). In Section~\ref{sec:properties} we give a behavioural equivalence between systems with failures, and  show how it is used to model properties of interest. Section~\ref{sec:app} describes prospective applications and promising directions of this work. Section~\ref{sec:related} discusses conclusions and related work.

\paragraph{Extensions with respect to the conference paper} This work is an extended version of the conference paper appeared in COORDINATION 2022~\cite{BLTV22} with the following additional contributions: 
\begin{itemize}
\item The syntax and semantics of systems have been extended to model checkpoints. We have added Section~\ref{reset} to show that the extension can still express the systems in~\cite{BLTV22} (Proposition~\ref{prop:reset}). In a new section, Section~\ref{sec:recoverability-cp}, we show that the notion of $n$-recoverability given in~\cite{BLTV22} is not suitable for systems with checkpoints. We have therefore added, in Section~\ref{sec:recoverability-cp}, a more relaxed notion of $n$-recoverability for checkpointing systems.
\item We found a flaw in the definition of $n$-recoverability given in~\cite{BLTV22}. We have fixed in this extended version. In Section~\ref{sec:recoverability-reset} we give an amended definition of $n$-recoverability alongside the one given in~\cite{BLTV22}, and discuss the differences by examples.
\item In~\cite{BLTV22} we informally stated a relationship between two reliability properties we defined in that work: `resilience is equivalent to $0$-recoverability'. In a new section of this extended version, Section~\ref{sec:res0rec}, we give the formal proof of this equivalence based on our amended definition of $n$-recoverability.
\item The examples in Section~\ref{sec:properties} have been improved to reflect the feedback from the presentation of the conference. For instance, Example~\ref{ex:resilience} has been framed to show the role of redundancy in fault-tolerance and how we can express it with our framework.
\item Following feedback at the conference presentation, we have extended the section of related works. In particular, we have added a comparison with the works in~\cite{10.1145/311531.311532} and~\cite{10.1007/s00165-017-0426-2}. Due to the particular relevance of the work in~\cite{10.1145/311531.311532}, we have also added a new subsection, Section~\ref{sec:fault}, with a more technical discussion on how our work relates to the more general definitions of fault-tolerance given in~\cite{10.1145/311531.311532}. 
\item We have added a Section~\ref{sec:app} with a discussion of prospective applications of this work.  
\end{itemize}

 \section{Informal overview}
\label{sec:informal}
\label{sec:overview}

Actor-based systems are modelled using a process calculus with three key elements, following the  actor model of Erlang: (1) time and timeouts, (2) asynchronous communication based on mailboxes with pattern-matching, and (3) actor nodes and injected failures.

\emph{Time and timeouts.} Timeouts are essential for an actor to decide when to trigger a recovery action. Time is also crucial to observe the effects of failure patterns including quantified delays or down-times of nodes and links.
We based our model of time on the Temporal Process Language (TPL)~\cite{TPL}, a well understood extension of CCS with discrete time and timeouts. Delays are processes of the form $\sleepk.P$ that behave as $P$ after one time unit. Timeouts are modelled after the idiomatic $\mathtt{receive ... after}$ pattern in Erlang.
Concretely, the Erlang pattern below (left) is modelled as the process below (right):

\begin{minipage}{.4\linewidth}
\begin{lstlisting}[style=erl]
receive
    Pattern1 -> P1;
    ...
    PatternN -> PN
after
    m -> Q
end
\end{lstlisting}
\end{minipage}
\begin{minipage}{.5\linewidth}
 \[?\{p_1.P_1, \ldots, p_N.P_N\}{\after{m} }\,Q\]
\end{minipage}

where $p_1,\ldots,p_N$ is a set of patterns, each associated with a continuation $P_i$, with $i \in \{ 1,\ldots,N\}$, and $Q$ is the timeout handler, executed if none of the patterns can be matched with a message in the mailbox within $m$ time units.
Following TPL, an action can be either a time action or an instantaneous communication action, and time actions can happen only when communication actions are not possible (maximal progress~\cite{TPL}). Concretely, we define the systems behaviour as a reduction relation with two kinds of actions: communication actions $\actarrow{}$ and time actions $\timearrow{}$. While TPL is synchronous and only prioritises synchronisations over delays, we model \emph{asynchronous} communications and prioritise any send or receive action over time actions. Thus, in our model, by maximal progress, communications have priority over delays.

The state of an actor at a time $t$ is modelled as $\node{n}{P}{\mb}{t}$, where $\nid n$ is the actor identifier (unique in the system), $\mb$ the mailbox, and $P$ the process run by that actor. System $\R_t$ below is the parallel composition of actors $\nid n_1$ and $\nid n_2$:
\[ \R_t=\node{n_1}{\sleepk.!\nid n_2\,\mathtt a. 0}{\emptyset}{t} \parallel \node{n_2}{?\mathtt{a}. P \after\, \color{red}1\color{black} ~  Q}{\emptyset}{t}\]
Although each actor in $\R_t$ has its own local time $t$ explicitly represented, which makes it easy to inject failures compositionally, our semantics keeps the time of parallel components synchronized (as in TPL). In $\R_t$, node $\nid n_1$ is deliberately idling and $\nid n_2$ is temporarily blocked on a receive/timeout action, so no communication can happen, and thus only a time action is possible, updating both actors' times and triggering the timeout in $\nid n_2$:
\[\R_t \timearrow{}
\node{n_1}{!\nid n_2\,\mathtt a. 0}{\emptyset}{t+1}\parallel\node{n_2}{Q}{\emptyset}{t+1} \]

\paragraph{Mailboxes.} Each pair of actors can communicate via two unidirectional links. For example, $(\nid n_1, \nid n_2)$ denotes the link for communications from $\nid n_1$ to $\nid n_2$.
An interaction involves three steps: (I) the sending actor sends the message by placing it in the appropriate link, (II) the message reaches the receiver's mailbox, and (III) the receiving actor processes the message. These three steps allows us to capture e.g., effects of failures in senders versus receivers, on nodes versus links, and to model latency. Consider the system \[\R_c =  \node{n_1}{!\mathtt{a}.0}{\emptyset}{t}\parallel\node{n_2}{?\mathtt{a}. P \after\, \color{red}2\color{black} ~  Q}{\mathtt{b}}{t}\]
Step (I), the sending of a message, is illustrated below on $\R_c$:
\begin{equation}\label{eqn:r}
	\R_c \actarrow{} \node{n_1}{0}{\emptyset}{t}\parallel \color{red}1\color{black}. (\nid n_1, \nid n_2, \mathtt a) \parallel\node{n_2}{ ?\mathtt{a}. P \after\, \color{red}2\color{black} ~  Q}{\emptyset}{t} =\R_c'
\end{equation}

\normalsize
$\color{red}1\color{black}.(\nid n_1, \nid n_2, \mathtt a)$ models a latent message in link $(\nid n_1, \nid n_2)$ with content $\mathtt a$. Prefix $1$ is the average network latency (assumed to be a constant). Due to latency, the message can only be added to the receiver's mailbox after one time step:
\begin{equation}\label{eqn:rc}
	\R_c' \timearrow{}\node{n_1}{0}{\emptyset}{t+1}\parallel (\nid n_1, \nid n_2, \mathtt a) \parallel\node{n_2}{ ?\mathtt{a}. P \after\, \color{red}1\color{black} ~  Q}{\emptyset}{t+1} 
\end{equation}
These floating messages $(\nid n_1, \nid n_2, \mathtt a)$ with no latency are similar to messages in the ether~\cite{SvenssonFE10}, in the global mailbox~\cite{LaneseNPV18}, or to the floating messages in~\cite{LaneseSZ19}.

Step (II) is the reception of the message, and happens as illustrated below (omitting the idle actor $\nid n_1$), where message $\mathtt a$ is added to the mailbox of $\nid n_2$:
\[\begin{array}{ll}
(\nid n_1, \nid n_2, \mathtt a) \parallel\node{n_2}{?\mathtt{a}. P \after\, \color{red}1\color{black} ~  Q}{\emptyset}{t+1} \actarrow{} \node{n_2}{?\mathtt{a}. P \after\, \color{red}1\color{black} ~  Q}{\mathtt a}{t+1}
\end{array}\]
Step (III) is the processing of the message, as illustrated below:
\[ \node{n_2}{?\mathtt{a}. P \after\, \color{red}1\color{black} ~  Q}{\mathtt a}{t+1} \actarrow{} \node{n_2}{P}{\emptyset}{t+1}\]
where message $\mathtt a$ in the mailbox matches the receive pattern (made up of a single atom $\mathtt{a}$) and is therefore processed.
Mailboxes give us an expressive model of communication for modern real-world systems. An alternative model of communication is  peer-to-peer communication, used e.g., in Communicating Finite State Machines (CFSM)~\cite{10.1145/322374.322380} and Multiparty Session Types~\cite{HondaYC16,DBLP:journals/mscs/CoppoDYP16}, where a receiver must specify from whom the message is expected. This makes it difficult to accurately capture interactions with public servers, or patterns like multiple producers-one consumer.

In the interaction above, note that $\nid n_2$ processes message $\mathtt a$ because it  matches pattern $\mathtt a$; this would be the case even if there were an older message $\mathtt b$ in the mailbox, if that message did not match that pattern $\mathtt a$. Alternative models, like Mailbox CFSMs~\cite{BasuBO12,BolligGFLLS21},
typically do not model the selective receive pattern (e.g., pattern-matching in Erlang) shown above. Without selective receive, participants can easily get stuck if messages are received out of order. One can encode peer-to-peer communication over FIFO unidirectional channels by using pattern matching with selective receive: using the sender's identifier in the message and in the receive pattern. A similar communication model to ours was proposed in~\cite{MostrousV11}.

\emph{Localities and failures.}
The actor construct is similar to that used to model locality for processes~\cite{Castellani01}, and also studied in relation to failures~\cite{BergerH00,RielyH01,DBLP:journals/jlp/FrancalanzaH07,DBLP:journals/iandc/FrancalanzaH08} but using a fail-stop untimed model. We use actor nodes to model the effects of injected failures on specific nodes and links.

Referring to system $\R_c'$ in~(\ref{eqn:r}), by placing floating messages into a link with latency before they reach the receiver's mailbox we can observe the effects of link failure as message loss. Assume link $(\nid n_1, \nid n_2)$ is down at time $t$:

\[\begin{array}{ll}
\R_c' \actarrow{}
 \node{n_1}{0}{\emptyset}{t}\parallel \node{n_2}{?\mathtt{a}. P \after\, \color{red}2\color{black} ~  Q}{\emptyset}{t}
\end{array}
\]
the floating message gets lost which in turn would end up causing a timeout in $\nid n_2$.
Similarly, in the case of node failure, node $\nid n_1$ in system $\R_c$, seen earlier in~(\ref{eqn:r}),  would go into a crashed node state before sending the message, hence triggering a timeout in $\nid n_2$:

\[\begin{array}{ll}
\R_c \actarrow{}
 \node{n_1}{\downarrow}{\emptyset}{t}\parallel \node{n_2}{?\mathtt{a}. P \after\, \color{red}2\color{black} ~  Q}{\emptyset}{t}
\end{array}
\]

\paragraph{Assumptions.} When a node crashes and comes back up again later on, it will come up with the same node identifier. This is consistent with Distributed Erlang, where by default  all nodes are named; on the other hand, if we were resuscitating processes, we would need to name them for this to be possible. For simplicity, we assume nodes are not created at run-time, focusing on fixed topologies. Extending the language with the capability of creating new nodes is relatively straightforward, and can be done in a similar way to $\pi$-calculus restriction.
We assume that behaviour within a node is sequential: actors can be composed in parallel but processes cannot, hence limiting communication to distributed communications between nodes. 

We choose to focus on inter-node communication on its own, because there already exist good strategies (e.g, in Erlang and Elixir) for dealing with in-node failure through the use of a supervision hierarchy, supervision strategies, and let-it-crash philosophy. Messages in transit when a node goes down remain in transit and may enter the mailbox after this node is resumed. 

We allow a restricted (external) version of choice, based on the communication patterns found in Erlang. Free, or completely unrestricted choice, while central to many process algebras, for example CCS, tends to be less used in practice.

\section{A model of failures}\label{sec:failures}

Let $\nodes$ be the set of node identifiers in a system. The model of failures is defined to be the $\Delta$ function:
\[\Delta: \nat \times (\nodes \cup \nodes \times \nodes) \mapsto \{ \csdown, \csup, \csslow\} \]
mapping each discrete time $t\in \nat$, node $\nid n \in  \nodes$, and link $(\nid n_1 , \nid n_2) \in \nodes \times \nodes$ to a value representing the state of health of that node or link, at that time. The symbol $\csup$ denotes the ``healthy'' state, 
$\csdown$ identifies the failure of a node or link, and $\csslow$ indicates a node or link slowdown. 

The failure scenarios covered by $\Delta$ include node crash, message loss, slow processes or slow networks. If \emph{node $\nid n$ is down} at time $t$, written $\curse{t}{\nid{n}} = \csdown$, then it will perform no action until it is resumed, if ever. If $\nid n$ is resumed at time $t'$, then its state at time $t'$ will be set to the initial state (see Definition~\ref{initial} for the formal definition). If \emph{link $(\nid n_1, \nid n_2)$ is down} at time $t$, written $\curse{t}{\nid{n}_1, \nid{n}_2} = \csdown$, then any message in transit on that link at time $t$ will be lost. If \emph{node $\nid n$ is slow} at time $t$, written $\curse{t}{\nid{n}} =\csslow$, then any actions of the process running in $\nid{n}$ are delayed for one time step, and may resume at time $t+1$ if $\curse{t+1}{\nid{n}} =\csup$.
If \emph{link $(\nid n_1, \nid n_2)$ is slow} at time $t$, written $\curse{t}{\nid{n}_1,\nid{n}_2} = \csslow$, then the delivery of any message in transit on that link at time $t$ will not happen at that time, and so will be delayed by at least one time unit. This delay is in addition to the network latency, which is modelled as a constant.
Failures can be permanent or transient, as shown below by examples.
\begin{exa}[Permanent and transient failures]
Permanent node failure after a certain point in time, say $t=10$, can be modelled by the definition $\Delta_1$ below. Function $\Delta_2$ shows a transient periodic structural failure of node $\nid{n}$, with each period having $100$ time units of healthy state and $100$ of down state. 
\[\ncurse{t}{\nid{n}}{1} = \begin{cases}\csup & \textit{if $t<10$}\\
\csdown & \textit{otherwise}\end{cases}
\qquad \qquad\quad
\ncurse{t}{\nid{n}}{2} = \begin{cases}\csup & \textit{if $t~ \mathtt{div} ~100~\mathtt{mod}~2~= 0$}\\
\csdown & \textit{otherwise}\end{cases}
\]
One could similarly model transient degrading failure by setting uptimes when $t=n^2$ for  $(n\in\nat)$.
\end{exa}

 \section{Calculus for cursed systems}\label{sec:calculus}
This section presents the model for actor based systems. The syntax of the calculus is given in~\Cref{fig:syntax}.

\begin{figure}[t]
	\begin{minipage}[t]{.45\linewidth}
		\[\arraycolsep=1.4pt\def\arraystretch{1}
		\begin{array}{lrll}
			\multicolumn{4}{l}{\text{Systems}} \\
			\R\!\!\!\!\! & \DEF & \nodecp{n}{P}{\mb}{\tval}{Q}        & \text{node}\\
			& \OR  & \PM{n_1}{n_2}{\msg}{\tval}               & \text{floating message} \\
			& \OR  & u.\PM{n_1}{n_2}{\msg}{\tval}             & \text{latent message} \\
			& \OR  & \nodecp{n}{\csdown}{\emptyset}{\tval}{Q}     & \text{crashed node}\\
			& \OR  & \emptyset                                & \text{empty}\\
			& \OR  & \R \pll \R                               & \text{parallel}\\
			\\\\
			\multicolumn{4}{l}{\text{Processes}} \\
			P & \DEF & \tsnd{\nid{n_i}}{m_i}{P_i}               & \text{send}\\
			& \OR  & \trcv{p_i.P_i}{\afterk {P}}              & \text{receive-timeout}\\
			& \OR  & \sleepk. P                               & \text{sleep}\\
			& \OR  & \savek. P                               & \text{check-point}\\
			& \OR  & \rec{t}P                                 & \text{fixed-point}\\
			& \OR  & \rvar{t}                                 & \text{recursive variable}\\
			& \OR  &\inact                                    & \text{inaction}\\
			\\
		\end{array}
		\]
	\end{minipage}
	\qquad \quad
	\begin{minipage}[t]{.45\linewidth}
		\[
		\begin{array}{lrll}
			\multicolumn{4}{l}{\text{Values}} \\
			\myval	& \DEF & a                                 & \text{atom}\\
			& \OR  & \nid{n}	                       & \text{node id}\\
			& \OR  & \myvar                            & \text{variable}
			\\
			\multicolumn{4}{l}{\text{Message}} \\
			\msg	& \DEF & \vec{V}			   & \text{message tuple}\\
			\\
			\multicolumn{4}{l}{\text{Mailbox}} \\
			\mb    & \DEF & \emptyset \OR \mb \cdot \msg      & \text{}\\
			\\
			\multicolumn{4}{l}{\text{Receive Patterns}} \\
			\ppar   & \DEF & \myvar  \OR a                     & \hspace*{-6pt}\text{pattern element}\\
			p       & \DEF & \vec \ppar                        & \hspace*{-6pt}\text{pattern tuple}\\
			\\
		\end{array}
		\]
	\end{minipage}
	\caption{Syntax}
	\label{fig:syntax}
\end{figure}

Systems are nodes $\nodecp{n}{P}{\mb}{\tval}{Q}$, messages (floating or latent), crashed nodes\linebreak[4]${\nodecp{n}{\csdown}{\emptyset}{\tval}{Q}}$, empty systems $\emptyset$, and parallel compositions of systems $\R \pll \R$. The term $\node{n}{P}{\mb}{\tval}$ denotes the state of node $\nid n \in \nodes$ at time $\color{red}t$ where $P$ is the process running in $\nid n$, $\color{gray} Q$ is the saved checkpoint process, and ${\color{blue}\mb}$ is the mailbox of $\nid n$. A {mailbox} is a (possibly empty) list of messages. A message $m$ is a tuple of values, which can be atoms $a$, node ids $\nid n$ or variables $X$. Messages are read from a mailbox via pattern matching. 
\begin{figure}[t]

\begin{align*}
	&\inferrule*[left=\VarA]{\match{\tvar}{a}{[a/\tvar]}}{}
&&
	\inferrule*[left=\VarN]{\match{\tvar}{\nid{n}}{[\nid{n}/\tvar]}}{}
\\[0.3cm]
	&\inferrule*[left=\Atom]{\match{a}{a}{\subst{a}{a}}}{}
&&
		\inferrule*[left=\Tuple]{
		\match{\ppar}{\myval}{\underline \sigma}
		\and
		\match{\vec \ppar}{\vec \myval}{\sigma}
	}{
		\match{\ppar \vec \ppar}{\myval \vec \myval}{\underline \sigma \sigma}
	}
	\\[0.3cm]
	&\inferrule*[left=\MboxH]{
		\match{\ppar}{\msg}{\sigma}
	}{
		\match{ \ppar}{\msg \cdot \mb}{ \sigma }
	}
	&&
	\inferrule*[left=\MboxT]{
		\nmatch{\ppar}{\msg}
		\and
		\match{ \ppar}{\mb}{\sigma}
	}{
		\match{ \ppar}{\msg \cdot \mb}{ \sigma}
	}
	\end{align*}
\caption{Matching rules}\label{match-fig}
\end{figure}

We define the pattern matching function in the style of \cite{MostrousV11} through the derivations in Figure \ref{match-fig}. Given a pattern $\vec \ppar$ and a message (tuple) $\vec \myval$, $\match{\vec \ppar}{\vec \myval}{\sigma}$ the match function returns a substitution $\sigma$. Note that the match is only defined if $\vec \ppar$ and $\vec \myval$ have the same size, and if the pattern and message match.
We write $\nmatch{\ppar}{\msg}$ when message $\msg$ does \emph{not} match pattern $\ppar$. 
Juxtaposition denotes concatenation of pattern and value tuples, and, since we assume that variables appear uniquely in pattern tuples, ${\underline \sigma} \sigma$ is the union of the two substitutions.

{A floating message} $\PM{{n_1}}{{n}_2}{\msg}{\tval}$ represents a message $\msg$ in link $(\nid{n_1},\nid{n_2})$. {Latent messages} $u.\PM{{n_1}}{{n}_2}{\msg}{\tval}$ are floating messages which can only reach the receiver's mailbox after a latency $u$. We assume all sent messages have a latency defined as a constant $L$, which abstracts the average network latency.

Looking at {processes}, a term of the form $\tsnd{\nid{n_i}}{m_i}{P_i}$ chooses to send to node $\nid{n_i}$ a message $m_i$ and continues as $P_i$. Term $\trcv{p_i.P_i}{\afterk {P}}$ tries to pattern match a message from the mailbox against one of the patterns $p_i$, and continues as $P_i$ given that the matching succeeds for $p_i$, timing out $\afterk$ one time unit if no message matches and executing $P$. Process $\sleepk.P$ consumes a time unit and then continues as $P$. Process $\savek.P$ saves the current state as a checkpoint process. Process $\mu \mathtt t . P$ is for recursion, and $\rvar{t}$ is the recursive call. Finally, $\inact$ is the idle process.
\begin{rem} We use notation $\trcv{p_i.P_i}{\after{u} {P}}$ as syntactic sugar for nesting $u$ timeouts\footnote{As $Q(u)$ where $Q(0) = \trcv{p_i.P_i}{\afterk P}$ and $Q(i+1) = \trcv{p_i.P_i}{\afterk Q(i)}$.}
and $\sleep{u}.P$  for the sequential composition of $u$ delays with continuation $P$.
\end{rem}

Recall (Section~\ref{sec:failures}) that we fix the set of system's nodes $\nodes$, and the domain of $\Delta$ is $\nodes \cup (\nodes \times\nodes)$, that is the set of nodes and links between pairs of nodes. Our unit of analysis is a \emph{cursed system} defined below.
\begin{defi}[Cursed system]\label{def_cs}
A cursed system is a pair $(\R,\Delta)$ where $\R$ is a system, $\Delta$ is a curse.
\end{defi}

The semantics of cursed systems is given in Def.~\ref{def_sem} as a reduction relation over  systems that is parametric on $\Delta$.
We write $\R_1 \equiv \R_2$ to mean that the systems $\R_1$ and $\R_2$ are the same up-to associativity and commutativity of~$\pll$, plus $0.\PM{{n_1}}{{n}_2}{\msg}{\tval}\equiv \PM{{n_1}}{{n}_2}{\msg}{\tval}$ and $\R \parallel \emptyset \equiv \R$.

\begin{defi}[Operational semantics for cursed systems]\label{def_sem}
Reduction is the smallest relation on cursed systems over communication actions denoted by $\actarrow{}$, and time actions denoted by $\timearrow{}$, that satisfies the rules in Figure~\ref{fig:reduction}. We use $\reduces{}$ when $\reduces{} \in \{\actarrow{}, \timearrow{}\}$.
For readability, in the rules we assume $\Delta$ fixed and write $\R \reduces{} \R'$ instead of $(\R, \Delta)\reduces{}(\R', \Delta)$.
\end{defi}

\begin{figure}\label{semrules}
	\begin{subfigure}[H]{\textwidth}
		\begin{align*}
& \inferrule*[left=\Snd]{
				\downf{\tk}{\nid{n}} = \csup
				\and 
				j \in I
			}{
				 \nodecp{n}{ \tsnd{\nid{n_i}}{m_i}{P_i}}{\mb}{t}{Q}
				\actarrow{}
				 \nodecp{n}{P_j}{\mb}{t}{Q} \pll L.\PM{n}{n_j}{\msg_j}{t}
			}
			\\[0.2cm]
&	\inferrule*[left=\Sched]{
				\downf{\tk}{\nid n_1} = \csup
				\and
				\downf{\tk}{\nid{n}_2, \nid{n}_1}  =  \csup
			}{
				\PM{n_2}{n_1}{\msg}{t} \pll \, \nodecp{n_1}{P}{\mb}{\tval}{Q}
				\actarrow{}
				\nodecp{n_1}{P}{\mb \cdot \msg}{\tval}{Q}
			}
			\\[0.2cm]
&	\inferrule*[left=\Rcv]{
				\downf{\tk}{\nid{n}} = \csup
				\and
				j\in I, \, \match{p_j}{\msg}{\sigma} \and
				\forall i\in I, \, \nmatch{p_i}{\mb_1}
			}{
				\nodecp{n}{\trcv{p_i.P_i}{\afterk P}}{\mb_1 \cdot \msg \cdot \mb_2 }{\tval}{Q}
				\actarrow{}
				\nodecp{n}{P_j \sigma}{\mb_1 \cdot \mb_2}{\tval}{Q}
			}
			\\[0.2cm]
&	\inferrule*[left=\Checkpoint]{
					-		}{
				\nodecp{n}{\savek.P}{\mb}{t}{Q}
				\actarrow{}
\nodecp{n}{P}{\mb}{t}{P}
			}\\[0.2cm]
&	\inferrule*[left=\Rec]{
				\downf{\tk}{\nid{n}} = \csup
				\and
				\nodecp{n}{P\subst{\rec{t}P}{\rvar{t}}}{\mb}{t}{Q}
				\reduces{}
\nodecp{n}{P'}{\mb}{t'}{Q}
			}{
				\nodecp{n}{\rec{t}P}{\mb}{t}{Q}
				\reduces{}
\nodecp{n}{P'}{\mb}{t'}{Q}
			}
		\end{align*}
		\caption{Actor/Node actions}\label{subfig:comm_actions}
		\noindent\rule{\textwidth}{0.5pt}
	\end{subfigure}
	\begin{subfigure}[H]{\textwidth}
		\begin{align*}
&	\inferrule*[left=\Sleep]{
				\downf{\tk}{\nid n} = \csup
			}{
				\nodecp{n}{\sleepk.P}{\mb}{\tval}{Q} \timearrow{} \nodecp{n}{P}{\mb}{\tval+1}{Q}
			}
			\\[0.2cm]
&	\inferrule*[left=\Latency]{
				\downf{\tk}{\nid{n_1}, \nid{n_2}} = \csup
				\and
				u > 0
			}{
				u.\PM{n_1}{n_2}{\msg}{\tval} \timearrow (u-1).\PM{n_1}{n_2}{\msg}{\tval+1}
			}
			\\[0.2cm]
&	\inferrule*[left=\Timeout]{
				\downf{\tk}{\nid{n}} = \csup \and \forall i \in I, \, \nmatch{p_i}{\mb}
			}{
				\nodecp{n}{\trcv{p_i.P_i}{\afterk P}}{\mb}{t}{Q} \timearrow
				\nodecp{n}{P}{\mb}{t+1}{Q}
			}
\end{align*}
		\caption{Time actions} \label{subfig:time_actions}
		\noindent\rule{\textwidth}{0.5pt}
	\end{subfigure}
	\begin{subfigure}[H]{\textwidth}
		\begin{align*}
			&    \inferrule*[left=\NLate]{
				\downf{\tk}{\nid{n}} = \csslow
			}{
				\nodecp{n}{P}{\mb}{\tval}{Q} \timearrow\nodecp{n}{P}{\mb}{\tval + 1}{Q}
			}
			&
\inferrule*[left=\MsgLoss]{
				\downf{\tk}{\nid{n_1}, \nid{n_2}} = \csdown
				\and
				u \geq 0
			}{
				u.\PM{n_1}{n_2}{\msg}{\tval} \actarrow{} \emptyset
			}
			\\[0.2cm]
&      \inferrule*[left=\MsgLate]{
				\downf{\tk}{\nid{n_1}, \nid{n_2}} = \csslow
				\and
				u \geq 0
			}{
				u.\PM{n_1}{n_2}{\msg}{\tval} \timearrow u.\PM{n_1}{n_2}{\msg}{\tval + 1}
			} 
			&
\inferrule*[left=\Down]{
				\downf{\tk}{\nid{n}}  = \csdown
			}{
				\nodecp{n}{P}{\mb}{\tval}{Q}  \actarrow{} \nodecp{n}{\csdown}{\emptyset}{\tval}{Q}
			}
			\\[0.2cm]
			&   \inferrule*[left=\NDLate]{
				\downf{\tk}{\nid{n}}  = \csdown
			}{
				\nodecp{n}{\csdown}{\emptyset}{\tval}{Q}  \timearrow{} \nodecp{n}{\csdown}{\emptyset}{\tval + 1}{Q}
			}
			&
\inferrule*[left=\Up]{
				\downf{\tk}{\nid{n}} = \csup 
			}{
				\nodecp{n}{\csdown}{\emptyset}{\tval}{Q}  \actarrow{}  \nodecp{n}{Q}{\emptyset}{\tval}{Q}
			}
		\end{align*}
		
		\caption{Failure actions}\label{subfig:failure_actions}
		\noindent\rule{\textwidth}{0.5pt}
	\end{subfigure}
	\begin{subfigure}[H]{\textwidth}
		\begin{align*}
			& \inferrule*[left=\Str]{
				\R_1 \equiv \R'_1 \and \R_1
				\reduces{}
				\R_2	\and \R_2 \equiv \R'_2
			}{
				\R'_1
				\reduces{}
				\R'_2
			}
			&&  \inferrule*[left=\ParCom]{
				\R_1
				\actarrow{}
				\R_1'
			}{
				\R_1 \pll \R_2
				\actarrow{}
				\R_1' \pll \R_2
			}
			\\[0.2cm]
&  \inferrule*[left=\ParTime]{
				\R_1
				\timearrow \R_1'
				\quad
				\R_2
				\timearrow \R_2'
				\quad
				\R_1\pll \R_2 \ncursearrow{}
			}{
				\R_1 \pll \R_2
				\timearrow
				\R_1' \pll \R_2'
			}
		\end{align*}
		\caption{System actions}\label{subfig:sys_actions}
		\rule{\textwidth}{0.5pt}
	\end{subfigure}
	
	\caption{Reduction rules}
	\label{fig:reduction}
\end{figure}

The first set of rules in~\cref{subfig:comm_actions} is for actors actions, happening at a time $\tval$, when the nodes and links are in a healthy state i.e. $\downf{t}{\nid n}=\csup$.
In {rule $\Snd$}, $\nid n$ chooses to send a message $m_j$ to node $\nid n_j$, and continues as $P_j$. Modelling asynchronous communication, a latent message $L.\PM{n}{n_j}{\msg_j}{t}$ is introduced in the system, where $L$ is the network latency constant.
{Rule $\Sched$} delivers a floating message to the receiver's mailbox. {Rule $\Rcv$}, retrieves the first message $\msg$ in the mailbox that matches one of the receive patterns $p_j$. The match function returns a substitution $\sigma$ that is applied to the continuation process $P_j$ associated with pattern $p_j$; and $\msg$ is removed from the mailbox. Rule $\Checkpoint$ saves the current state $P$ as a checkpoint process for that node $\nid n$. Finally, {Rule $\Rec$} allows a node with a recursive process to proceed with a communication or a time action.

\paragraph{Time actions.}
The second set of rules, in~\cref{subfig:time_actions}, is for time-passing reduction in absence of failures. {Rules $\Sleep$ and $\Timeout$} model reduction of time consuming and receiving with timeout processes, respectively. {Rule  $\Timeout$} can only be applied if none of the messages in the mailbox is matching any of the patterns $\{p_i\}_{i\in I}$ yielding an urgent receive semantics~\cite{MURGIA201938} reflecting the receive primitive in Erlang. {Rule $\Latency$} allows time passing for latent messages. Note that, by setting $u' = \mathtt{max}(u-1, 0)$, if a receiver node crashes, all latent/floating messages remain in the link until the node is able to receive them, i.e. in a healthy state. We omit the rules for state-preserving time passing for idle nodes and $\node{n}{\inact}{\mbox{M}}{\tval}$.

\paragraph{Failure actions.}
The third set of rules, in~\cref{subfig:failure_actions}, models the effects of failures injected at time $t$.
{Rule  $\NLate$} models a delay, injected by $\down(\tk)(\nid{n}) = \csslow$, in the execution of the process $P$ in a node $\nid n$: a time unit elapses without any action in $P$. {Rule $\MsgLoss$} models a lossy link at time $\tval$, injected by $\down(\tk)(\nid{n_1}, \nid{n_2}) = \csdown$, and permanently deletes a message $u.\PM{n_1}{n_2}{\msg}{\tval}$ in transit. {Rule $\MsgLate$} models a slow link, injected by $\down(\tk)(\nid{n_1}, \nid{n_2}) = \csslow$, by allowing time to pass but without decreasing the latency $u$ of the message. {Rule $\Down$} models an instantaneous node that crash injected by $\down(\tk)(\nid{n}) = \csdown$, and erases the process and mailbox of the node. {Rule $\NDLate$} allows time to pass for a crashed node. In rule $\Up$ a crashed node is restarted with its saved checkpoint process $Q$ and empty mailbox. $\pf$ is a mapping from $\nodes$ to processes, that gives the initial process of each actor node. We assume that the node identifier is unchanged when restarting the node.

\paragraph{Runtime System actions.}
The last set of rules given in~\cref{subfig:sys_actions} models system actions.
In {rule $\ParCom$} a communication action of system part $\R_1$ is reflected in the composite system $\R_1 \pll \R_2$.
In {rule $\ParTime$} time actions need to be reflected in all the parts of a system. A whole system can have a time action only if all parts of the system have no communication or failure actions to perform at the current time ($\R_i \nactarrow{}$). $\Str$ is for communication and time actions of structurally equivalent systems.

\subsection{Basic properties of systems reductions}
In the remainder of this section we discuss two properties of cursed systems: time-coherence (the semantics keeps clocks synchronized) and non-Zenoness. We start by defining the time of a system. All definitions below apply straightforwardly to cursed systems by fixing a $\Delta$.

\newcommand{\tw}{\underline{t}}
\newcommand{\stime}{\mathtt{time}}

\begin{defi}[Time of a system] Let $\tw$ range over $\nat \cup \{*\}$. We define the synchronization (partial) function $\coherent$:
\[ \coherent(*, \tw) =  \coherent(\tw, *) = \tw \quad \coherent(*, *) = * \quad  \coherent(\tw, \tw)=\tw  \]
$\coherent(\tw_1, \tw_2)$ returns a time or a wildcard $*$, and is undefined if $\tw_1\not =\tw_2$ and neither $\tw_1$ nor $\tw_2$ is a wildcard.
We define $\stime(\R)$ as a partial function over systems:
\[\stime(\R) = 	\begin{cases}
					* &  \R = \emptyset \\
					t  & \R = \nodecp{n}{P}{\mb}{t}{Q} \text{ or } \R = \nodecp{n}{\csdown}{\mb}{t}{Q} \text{ or }\\
					  & \R = \PM{n_1}{n_2}{m}{t}\text{ or } \R = u.\PM{n_1}{n_2}{m}{t}\\
					\coherent(\stime(\R_1), \stime(\R_2 ))& \R = \R_1 \pll \R_2
					\end{cases}
		\]
	\end{defi}

 We can now define time-coherence of a system, holding when all its components have the same time.
\begin{defi}[Time coherence]
$\R$ is \emph{time coherent} if $\stime(\R)$ is defined.
\end{defi}
For example, system $\nodecp{n_1}{P}{\mb}{t}{Q_1}\parallel\PM{n_1}{n_2}{\msg}{t}\parallel \emptyset$ is time-coherent, while system $\nodecp{n_1}{P}{\mb}{t}{Q_1}\parallel \PM{n_1}{n_2}{\msg}{t+1}\parallel \emptyset$ is not.

The time function is also useful to characterise systems where all actors are coherently at time $0$ and in their initial state.
\begin{defi}[Initial system]\label{initial}
Let $\pf$ and $\cpf$ be mappings from $\nodes$ to processes such that $\pf(\nid n)$ is the initial process of $\nid n$ and $\cpf(\nid n)$ is the initial checkpoint of $\nid n$. Note that by the definition of processes $\downarrow$ is not a process, and so nodes are never crashed in the initial state. A system $\R$ is initial if $\stime(\R) = 0$ and
	\[\R \equiv \nodecp{n_1}{\pf(\nid n_1)}{\emptyset}{0}{\cpf(\nid n_1)} \pll \ldots \pll \nodecp{n_m}{\pf(\nid n_m)}{\emptyset}{0}{\cpf(\nid n_m)}\]
	with $\{1,\ldots,m\} = \nodes$.  A cursed system $(\R,\Delta)$ is \emph{initial} if $\R$ is initial.
\end{defi}

We assume any system $\R$ to start off as initial and hence, by Prop.~\ref{time-coherence}, to be time-coherent.

Next we show that the reduction over systems preserves time-coherence, hence all reachable systems are coherent.

\begin{lem}[Time-coherence invariant]\label{time-invariant}
If $\R$ is time-coherent and $\R \reduces{} \R'$ then $\R'$ is time-coherent.
\end{lem}
The proof of the lemma is straightforward, by induction on the derivation. In fact, the only rule that updates the time of a parallel composition is $\ParTime$ which requires time passing for all parallel processes. The fact that if $\R$ is initial then $time(\R)$ is defined (as $0$) yields the following property. We let $\reducesStar{}$ be the transitive closure of the reduction relation.

\begin{prop}\label{time-coherence}
Let $\R$ be initial, if $\R\reducesStar{}\R'$ then $\R'$ is time-coherent.
\end{prop}
We assume any system $\R$ to start off as initial and hence, by Prop.~\ref{time-coherence}, to be time-coherent.

Next, we give a desirable property for timed models: non-Zenoness. This prevents an infinite number of communication actions at any given time (Zeno behaviours). Besides yielding a more natural abstraction of a real world system, non-Zenoness simplifies analysis; for example, we can assume that the set of states reachable without time passing is finite. We start by defining a non-instantaneous process.

\begin{defi}[Non-instantaneous process]
We define function $\ninsta(P)$ inductively as follows:

\[\ninsta(P)\! =\! \begin{cases}
	\! \bigwedge_{i\in I} \ninsta(P_i) & \!\!\!  \text{if $P = \tsnd{\nid{n_i}}{m_i}{P_i}$ or $P = \trcv{p_i.P_i}{\afterk\, {Q}}  $} \\
	\! \ninsta(Q) & \!\!\!  \text{if $P = \mu \tvar . Q$}\\
	\! \true & \!\!\!  \text{if $P = \sleep . Q$}\\
	\! \false & \!\!\!  \text{if $P = \tvar$ or $P=0$}
\end{cases}
\]
We say that $P$ is non-instantaneous if $\ninsta(P) = \true$. We say that $\R$ is non-instantaneous if all nodes in $\R$ run non-instantaneous processes.
\end{defi}

\begin{prop}[Non-Zenoness]\label{nonzeno}
Let $\R$ be non-instantaneous. If $\R\reducesStar{}\R'$ then there is a finite number of $\R''$ such that $\R'\actarrow{}\R''$.
\end{prop}
The proof is straightforward by induction on the structure of $\R'$. Intuitively, any non-instantanous actor can only make a finite number of instantaneous actions at any given time, and hence at time $\stime(\R')$.
Hereafter we assume systems to be non-instantaneous, and hence non-Zeno.

\subsection{Reset vs Checkpointing Systems}
\label{reset}

We call \emph{reset systems} those systems obtained using the grammar for systems but without the save processes $\savek.P$, and where $\Gamma = \Sigma$. Reset systems model systems where each node reacts to (presumed) failure by restarting the execution from the initial state. More formally: 

\begin{prop}[Reset systems]\label{prop:reset}
If $\R$ is reachable from an initial reset system then for all $\nodecp{n}{P}{\mb}{t}{Q}$ and $\R'$ 
such that $\R = \nodecp{n}{P}{\mb}{t}{Q}\pll \R'$ we have $Q=\pf(\nid n)$.
\end{prop}
The property above is proved straightforwardly by coinduction, showing that having checkpoint $\pf(\nid n)$ in all nodes is a property of initial reset systems and an invariant of reset systems preserved by reduction (by case analysis on the reduction rules).  

Reset systems are common in Erlang: robustness is provided by a supervision hierarchy which explicitly describes the ways in which parts of the system are restarted when they or other parts fail. Restarts can escalate: if a component repeatedly restarts, then its parent process may itself have to be restarted. 

While Erlang provides no explicit mechanism for checkpointing, it is possible to save state periodically using bulk storage known as ETS-tables. These provide global storage from which state can be retrieved, always assuming that the tables themselves are preserved. Disk-based ETS-tables (DETS-tables) provide more permanent storage, but with an associated time cost.

In fact, in the short version of this article \cite{BLTV22} we focussed on a formalisms that corresponds to reset systems. 
Here, we explore a more general setting, to show a more interesting relationship of our work to the ones in \cite{10.1145/311531.311532}\cite{10.1007/s00165-017-0426-2}, and particularly to the notion of \emph{non-masking} fault-tolerance therein.

 \section{Properties of cursed systems}
\label{sec:properties}

In this section we define a behavioural relation between cursed systems, as a weak barbed bisimulation, which is the standard choice since we have a reduction semantics~\cite{SW01}. The aim is to compare the systems' abilities to preserve `normal' functionality when they are affected by failures. We abstract from the fact that some parts of the system may be deadlocked, as long as healthy actors can keep receiving the messages they expect. Mailbox-based (rather than point-to-point) communication and pattern matching allow us to capture e.g., multiple-producer scenarios where a consumer can receive the expected feeds as long as \emph{some} producers are healthy. 

Our behavioural relation also abstracts from time, to disregard the delays introduced by recovering actions, and only observes the effects of such delays (we do not focus on efficiency). Essentially, two systems are equivalent when actors receive the same messages, abstracting from senders, in a time-abstract way.\footnote{Abstracting from timing and message senders is an assumption of our model that we adopted for the sake of generality: it allows us to capture scenarios where the timing and order of the messages does not matter (e.g., multiple producers). On the other hand the model can encode scenarios where such orders matter. It can, for example, support Erlang-style actor behaviour. Erlang does not guarantee temporal order of messages between different processes in general, however between any two processes it does guarantee that messages sent directly between them will be received in the same order. Erlang behaviour can be encoded in our model if messages are extended to include the identity of the sender and a counter (e.g., as atoms) to guarantee message origin and ordering.} On the basis of this equivalence we define \emph{recoverability} and \emph{augmentation}.

We start by defining weak barbed simulation for cursed systems.
\newcommand{\mybarb}[1]{\downarrow\, {#1}}
\begin{defi}[Barb]\label{barb}
The ready actions of $P$ are defined inductively as follows:
\[\begin{array}{lll}
\ready{  \tsnd{\nid{n}_i}{m_i}{P_i}} =  \{ \snd \nid{n}_i\, m_i \}_{ i\in I}  & \quad
\ready{\trcv{p_i.P_i}{\after{}{ P }}} = \{\rcv p_i \}_{i\in I} \\
\ready{0}=\ready{\mathtt t} = \ready{\sleepk.P} = \emptyset  & \quad
 \ready{\mu \mathtt t. P} = \ready{P}
\end{array}\]

Let $\R \downarrow \brb$ be the least relation satisfying the rules below.
\[\begin{array}{lll}
\node{n}{P}{\mb}{t} \mybarb{x} & \quad \text{if $ \snd \nid{n}' m \in\ready{P} \land x = \snd \nid{n}' m~\lor~\rcv p \in \ready{P} \land x = \rcv \nid n\, p$}\\

(\nid{n}_1,\nid{n}_2,m) \mybarb{\snd \nid{n_2}\, m} \\
(\R_1 \parallel \R_2)\mybarb{\brb} &   \quad  \text{if $\R_1\mybarb{\brb}$ or   $\R_2\mybarb{\brb}$}\\
\end{array}\]
If $\R\mybarb{\brb}$ we say that $\R$ has a \emph{barb} on $\brb$.
\end{defi}

Barbs abstract from (i.e., do not include in the model of observation) the sender of a message. This allows us to disregard the identity of the senders, following mailbox-based communications in actor-based systems. Scenarios where the identity of the sender is important can be encoded by using node identifiers as message content.\footnote{This is precisely how sender information is communicated in Erlang.} We observe $m$ and $p$ to retain expressiveness with respect to channel-based scenarios, as discussed in Section~\ref{scopes}.

\begin{exa}[Examples on barbs]\label{ex:barbs}
Consider a system $\R_{R}$ with a consumer node $\mathtt c$ receiving data $\mathtt d$ from two replicas $\mathtt r1$ and $\mathtt r2$. If the messages from both replicas are delayed then the consumer notifies a monitor node $\mathtt m$ (omitted here for simplicity): 
\[\begin{array}{lll}
\R_{R} =    \node{c}{\mu \mathtt t. ?\mathtt{d}.  \mathtt t. \after\, \color{red}2\color{black} ~  !\mathtt{m}\,\, \mathtt{fail}  }{\emptyset}{0}\parallel
\node{r1}{\mu \mathtt t.!\mathtt{c}\,\mathtt{d}.\mathtt t}{\emptyset}{0} \parallel \node{r2}{\mu \mathtt t.!\mathtt{c}\,\mathtt{d}.\mathtt t}{\emptyset}{0}
\end{array}
\]

Regarding our choice of barbs in this example, the consumer needs to receive regular feeds $\mathtt d$, no matter whether they are from $\mathtt{r1}$ or $\mathtt{r2}$. Abstracting away from the identity of the sending replica is directly captured by our definition of barbs. In fact, the set of barbs of $\R_{R}$ is $\{\snd \nid{c}\, \mathtt d, \, \rcv \nid{c}\, \mathtt d\}$. 

A system defined in the same way as $\R_R$ but with only one replica, e.g.\ obtained by removing node $\mathtt{r2}$, or with one of the replicas down, e.g.\ obtained by substituting node $\mathtt{r2}$ with $ \node{r2}{\downarrow}{\emptyset}{0}$, would have the same set of barbs as $\R_R$, namely $\{\snd \nid{c}\, \mathtt d, \, \rcv \nid{c}\, \mathtt d\}$. 

It is worth noting that if the identity of the sender does matter, it can be observed by encoding the identity into the messages being sent by the sender: 
  \[\begin{array}{lll}
\R'_{R} = &    \node{c}{\mu \mathtt t. \{ ? \mathtt{n1}\, \mathtt{d}.  \mathtt t., \, ? \mathtt{n2}\, \mathtt{d}.  \mathtt t. \after\, \color{red}2\color{black} ~  !\mathtt{m}\,\, \mathtt{fail}  }{\emptyset}{0}\parallel \\
& \node{r1}{\mu \mathtt t.!\mathtt{c}\,\mathtt{n1}\, \mathtt{d}.\mathtt t}{\emptyset}{0} \parallel \node{r2}{\mu \mathtt t.!\mathtt{c}\,\mathtt{n2}\, \mathtt{d}.\mathtt t}{\emptyset}{0}
\end{array}
\]
The set of barbs of $\R'_{R} $ is $\{\snd \nid{c}\, \mathtt {n1} \, \mathtt d, \snd \nid{c}\, \mathtt {n2} \, \mathtt d, \, \rcv \nid{c}\, \mathtt {n1} \,\mathtt d, \, \rcv \nid{c}\, \mathtt {n2} \,\mathtt d\}$. If node $\mathtt{r2}$ was removed or crashed, the set of barbs would be affected, becoming $\{\snd \nid{c}\, \mathtt {n1} \, \mathtt d,  \, \rcv \nid{c}\, \mathtt {n1} \,\mathtt d, \, \rcv \nid{c}\, \mathtt {n2} \,\mathtt d\}$, and making it possible to distinguish among senders. 
\end{exa}

\noindent
\begin{defi}[Weak barbed simulation]\label{simulation}
Recall $\reduces{}\in \{\actarrow{}, \timearrow\}$.
A weak (time-abstract) barbed simulation is a binary relation $\simu$ between cursed systems such that $(\R_1,\Delta_1) \simu (\R_2,\Delta_2)$ implies:
\begin{enumerate}
\item If $(\R_1,\Delta_1)\reduces{} (\R_1',\Delta_1)$ then there exists $\R_2'$ such that $(\R_2,\Delta_2) \reducesStar{}(\R_2',\Delta_2)$ and $(\R_1',\Delta_1) \simu (\R_2',\Delta_2)$.
\item If $\R_1 \mybarb{\brb}$ for some $\brb$, then there exists $\R_2'$ such that $(\R_2,\Delta_2) \reducesStar{} (\R_2',\Delta_2)$ and $\R_2'\mybarb{\brb}$.
\end{enumerate}

\noindent We say $(\R_1,\Delta_1)$ is weak barbed similar to $(\R_2,\Delta_2)$, written $(\R_1,\Delta_1) \lesssim (\R_2,\Delta_2)$, if there exists some weak barbed simulation $\simu$ such that $(\R_1,\Delta_1) \simu (\R_2,\Delta_2)$. 
\end{defi}

\begin{defi}[Weak barbed bisimulation]\label{bisimulation}
We say that $\simu$ is a weak barbed bisimulation if $\simu$ and $\simu^{-1}$ are weak barbed simulations.
We say $(\R_1,\Delta_1)$ is weak barbed bisimilar to $(\R_2,\Delta_2)$, written $(\R_1,\Delta_1) \approx (\R_2,\Delta_2)$, if there exists some weak barbed bisimulation $\simu$ such that $(\R_1,\Delta_1) \simu (\R_2,\Delta_2)$. 
\end{defi}

The rules in Figure \ref{fig:reduction} embody mailbox-based communication. These rules allow us to observe messages `in flight', and so allow us to observe that some such messages have been affected by a curse, reflecting an insecure communication medium. Note, however, that it is \emph{not} possible to directly observe the contents of mailboxes; this can only be done indirectly -- and in general, partially -- by observing the behaviour of receive statements.

If we fix a system, we can use Definition \ref{simulation} to compare behaviours of that system with different curses, namely to determine when a system maintains its (desirable) behaviour even when cursed. 
Namely, Definition \ref{simulation} provides a means to study the ability of a systems to tackle failure, or its \emph{resilience}. Formally (Definition \ref{def:resilience}) we define resilience as the ability of a system to behave `normally' despite failure injection. In the following, we let $\uparrow$ be the curse function that assigns $\uparrow$ to all nodes and links at all times.

\begin{defi}[Resilience]\label{def:resilience}
Initial $(\R,\Delta)$ is resilient if \( (\R,\uparrow)\approx (\R,\Delta)\).
\end{defi}

\begin{exa}[Resilience]\label{ex:resilience}
Consider the system $\R$ below and curse $\Delta$ for which $(\mathtt{p},\mathtt{c})$ is down at time $1$ 
\[\begin{array}{lll}
\R =   \node{p}{\sleep{}.!\mathtt{c}\,\mathtt{item}.0}{\emptyset}{0}\parallel\node{c}{?\mathtt{item}. 0 \after\, \color{red}3\color{black} ~  0}{\emptyset}{0} 
\end{array}
\]
Fix the latency constant as $1$ time unit. System $\R$ is not resilient with respect to $\Delta$ since $(\R,\Delta)\not \lesssim(\R,\uparrow)$. Intuitively, observe that $(\R,\uparrow)$ reaches the terminated state 
\[
(  \node{p}{0}{\emptyset}{2}\parallel\node{c}{?\mathtt{item}. 0 \after\, \color{red}1\color{black} ~  0}{\mathtt{item}}{2} ,\uparrow)\]
with no barbs, whereas $(\R,\Delta)$ finally gets stuck in the state 
\[(  \node{p}{0}{\emptyset}{3}\parallel\node{c}{0}{\emptyset}{3}\parallel (\nid{p},\nid{c},\mathtt{item}),\Delta) \] 
with an orphan message and barb $?\nid{c}\,\mathtt{item}$. 

Different strategies can be applied to modify $\R$ so that it correctly handles the delays specified by $\Delta$. 
One is to tune the timeouts in the code so that it can handle the curse $\Delta$. Concretely, consider the following variant $\R'$ of $\R$, that increases the timeout value of one unit: 
\[
\R' =  \node{p}{\sleep{}.!\mathtt{c}\,\mathtt{item}.0}{\emptyset}{0}\parallel\node{c}{?\mathtt{item}. 0 \after\, \color{red}4\color{black} ~  0}{\emptyset}{0}
\]
One can verify that $(\R',\Delta) \approx(\R',\uparrow)$ and hence $\R'$ is resilient with respect to $\Delta$.
\end{exa}

Example~\ref{ex:resilience} shows a non-resilient cursed system $(\R, \Delta)$ and a resilient variant $(\R',\Delta)$ obtained by tuning the timeout in $\R$. In the following Example~\ref{ex:redundance} we provide two additional resilient variants of $(\R, \Delta)$ obtained using time-redundancy (e.g., retry strategies) and space-redundancy (e.g., replication). Redundancy has been shown~\cite{10.1145/311531.311532} to be a necessary condition for fault-tolerance. Resilience gives a tool to assess whether a `redundant' system is indeed attaining the intended fault-tolerance.  

\begin{exa}[Resilience and redundancy]\label{ex:redundance}
Consider $(\R,\Delta)$ from Example~\ref{ex:resilience} and, again, fix the latency constant as $1$ time unit. We define a variant of $\R$, called $\R_T$, where time-redundancy is attained by retrying the communication once more in case of timeout
\[
\R_T =   \node{p}{\sleep{}.!\mathtt{c}\,\mathtt{item}.0}{\emptyset}{0}\parallel\node{c}{?\mathtt{item}. 0 \after\, \color{red}3\color{black} ~  (?\mathtt{item}. 0 \after\, \color{red}3\color{black} ~  0)}{\emptyset}{0}
\]
Similarly, we define a variant of $\R$, called $\R_S$, where space-redundancy is applied by adding an extra producer: 
\[\begin{array}{ll}
\R_S =  & \node{p}{\sleep{}.!\mathtt{c}\,\mathtt{item}.0}{\emptyset}{0}\parallel \node{c}{?\mathtt{item}. 0 \after\, \color{red}3\color{black} ~  0 }{\emptyset}{0}\\
& \parallel\node{p'}{\sleep{}.!\mathtt{c}\,\mathtt{item}.0}{\emptyset}{0}
\end{array}\]
One can verify that both $\R_T$ and $\R_S$ are resilient with respect to $\Delta$ from Example \ref{ex:resilience}. 
\end{exa}

Our definition of resilience sets the behaviour of a system without curses as a model of \emph{expected behaviour}. 
By Definition~\ref{def:resilience}, any deviation from the expected behaviour, even a temporary one, makes a system non-resilient. This is a very strict characterization of fault-tolerance. For example, resilience is too strong to capture the effects of more complex retry-strategies than those applied in $\R_T$ from Example~\ref{ex:redundance}, as shown in the Example~\ref{ex:retry} below.

\begin{exa}[Resilience and more complex retry strategies]\label{ex:retry}
Consider $\Delta$ from Example~\ref{ex:resilience}, latency of $1$ time unit, and a variant $\R_{TT}$ of $\R_T$, where time-redundancy affects both processes: 
\[\begin{array}{lll}
\R_{TT} = &   \node{p}{\mu \tvar. \, \sleep{}.!\mathtt{c}\,\mathtt{item}.\, ? \{ \mathtt{ok}\, . 0, \, \mathtt{retry}. \tvar\}}{\emptyset}{0}\parallel\\
& \node{c}{\mu \tvar.\,  ?\mathtt{item}.\ !\mathtt{p}\,\mathtt{ok} .\ 0 \after\, \color{red}3\color{black} ~ ( !\nid{p}\, \mathtt{retry}.\, \tvar)}{\emptyset}{0}
\end{array}\]
System $\R_{TT}$ is not resilient with respect to $\Delta$ from Example \ref{ex:resilience} because the nodes add some communications to acknowledge correct interaction or coordinate on a retry iteration. 
\end{exa}

In the remaining of this section, we study a less restrictive characterization than resilience, which we call recoverability, to allow for some deviation from the expected behaviour as long as the system \emph{eventually} resumes the expected behaviour. In Section~\ref{sec:recoverability-reset} we discuss recoverability.
In Section~\ref{sec:res0rec} we show a relation between resilience and recoverability.
In Section~\ref{sec:recoverability-cp} we provide a more general account of reliability that can easily capture reset and checkpointing systems. Section~\ref{sec:recoverability-reset} is based on the notion of $n$-recoverability first introduced in \cite{BLTV22}, which is fixed and improved.
Section~\ref{sec:recoverability-cp} is new. 

\subsection{Recoverability for reset-systems}
\label{sec:recoverability-reset}

We define $n$-recoverability as the ability of a system to display the expected behaviour after time $n$. The definition from \cite{BLTV22} had several issues that we have amended in this work. The original definition is as follows: 

\begin{defi}[$n$-Recoverability (from~\cite{BLTV22})]\label{recoveraold}
Let $n\in\nat$ and $(\R,\Delta)$ initial. $(\R,\Delta)$ is \emph{$n$-recoverable} 
if $(\R,\Delta)\reducesStar{}  (\R',\Delta)$ and  $\stime(\R')=n$,  implies $(\R,\uparrow) \approx (\R',\Delta)$.
\end{defi}

\begin{exa}[Counterexample]\label{ex:counter}
Fix a latency of $1$ time unit and a generic $\Delta$ that does not affect any node or link at time $0$. 
System $\R_{ce}$ below reduces to $\R_{ce}'$ after a communication action 
\[\begin{array}{lll}
\R_{ce} = &   \node{n1}{!\nid{n2}\,\mathtt{a}. 0}{\emptyset}{0}\parallel
 \node{n2}{\sleep{}.  ?\mathtt{a}. 0}{\emptyset}{0} \\
\R_{ce}' =  & \node{n1}{0}{\emptyset}{0}\parallel
 \node{n2}{\sleep{}.  ?\mathtt{a}. 0}{\emptyset}{0}\parallel \mathbf{1}.\PM{{n_1}}{{n}_2}{\mathtt a}{0}
 \end{array}\]
For $\R_{ce}$ to be $0$-recoverable, since  $(\R_{ce},\Delta)\actarrow{}(\R_{ce}',\Delta)$, one should have $(\R_{ce}', \Delta)\approx (\R_{ce}, \uparrow)$ which does not hold for any $\Delta$ because of a difference in barbs-- hence not even for $\Delta = \uparrow$.
\end{exa}

Example~\ref{ex:counter} shows that Definition~\ref{recoveraold} is too strict to capture the intended meaning of $n$-recoverability. In~\cite{BLTV22} for example, $0$-recoverability is (wrongly) set to correspond to resilience. Definition~\ref{recoveraold} requires that \emph{all} states at time $n$ are bisimilar to the initial state, and this is too strict since several actions may naturally happen in a time unit. 

We provide a weaker definition of $n$-recoverability, using universal quantification over \emph{paths} of actions at time $n$ and existential quantification on the states on each of these paths, which is set to better represent the intuition. First, we define the concept of $n$-entry, which is the set of states that are, for some execution, the first state to be reached at time $n$. Then a $n$-path is the maximal path from a $n$-entry where states are at time $n$.

\begin{defi}[$n$-Entry]\label{nentry}
Let $n\in\nat$ and $(\R_0,\Delta)$ be an initial state. If $n=0$ then $(R_0,\Delta)$ is the only $0$-entry for itself.
If $n>0$, $(\R,\Delta)$ is a $n$-entry for $(\R_0,\Delta)$ if there exists an execution $(\R_0,\Delta)\reducesStar{} (\R',\Delta)\rightsquigarrow(\R,\Delta)$ with $\stime(\R)=n$.
\end{defi}
A $n$-entry $(\R,\Delta)$ is the first state to be reached at time $n$. Observe that in Definition~\ref{nentry} if $n>0$ then it is always the case that $\stime(\R')=n-1$. We define an \emph{execution} $(\R_1,\Delta)\rightharpoonup^*(\R_m,\Delta)$ to be a sequence of configurations $(\R_i,\Delta)$, with $1 \leq i \leq m-1$ such that $(\R_i,\Delta)\rightharpoonup(\R_{i+1},\Delta)$.

\begin{defi}[$n$-Path]\label{nparth}
Let $n\in\nat$ and $(\R_0,\Delta)$ be an initial state. 
Execution $(\R_1,\Delta)\rightharpoonup^*(\R_m,\Delta)$ is a $n$-path for $(\R_0,\Delta)$ if: (1)  $(\R_1,\Delta)$ is a $n$-entry for $(\R_0,\Delta)$, and (2) $(\R_m,\Delta)$ cannot make other actions than time actions. 
\end{defi}

Observe that in Definition~\ref{nparth} a $n$-path $(\R_1,\Delta)\rightharpoonup^*(\R_m,\Delta)$ includes no time actions. 

\begin{defi}[$n$-Recoverability (new)]\label{recoveranew}
Let $n\in\nat$ and $(\R_0,\Delta)$ be an initial state. $(\R_0,\Delta)$ is $n$-recoverable if for all of its $n$-paths 
$(\R_1,\Delta)\rightharpoonup^*(\R_m,\Delta)$ there is $i\in\{1,\ldots,m\}$ such that $(\R_i,\Delta)\approx(\R_0,\uparrow)$.
\end{defi}

Definition~\ref{recoveranew} says that in any arbitrary $n$-path there exists a state $(\R_i,\Delta)$ reachable at time $n$ that is weak-barbed bisimilar to $(\R_0,\uparrow)$.

\begin{exa}[$n$-Recoverability]\label{ex:nrec}
Consider the system $\R$ below (and any $\Delta$ that does not  affect the system at times $0$ and $1$): 
\[
\R =  \node{n_1}{!\mathtt{n_2}\,\mathtt{a}.\emptyset}{\emptyset}{0}\parallel\node{n_2}{?\mathtt{a}. \emptyset}{\emptyset}{0}
\]
$(\R, \Delta)$ reduces to the successfully terminated system $(\R', \Delta)$ with 
\[\R' = \node{n_1}{\emptyset}{\emptyset}{1}\parallel\node{n_2}{\emptyset}{\emptyset}{1}\] 
at time zero.
\end{exa}

\begin{exa}[$n$-recoverability and more complex retry strategies]\label{ex:retryn}
	Consider $\Delta$ from Example~\ref{ex:redundance}, latency of $1$ time unit, and a variant $\R_{TT}$ of $\R_T$, where time-redundancy affects both processes: 
	\[\begin{array}{lll}
		\R_{TT} = &   \node{p}{\mu \tvar. \, \sleep{}.!\mathtt{c}\,\mathtt{item}.\, ? \{ \mathtt{ok}\, . 0, \, \mathtt{retry}. \tvar \} \after\,{\color{red}5}~ 0}{\emptyset}{0}\parallel
		\\
		& \node{c}{\mu \tvar.\,  ?\mathtt{item}.!\mathtt{p}\,\mathtt{ok}. 0 \after\, \color{red}4\color{black} ~ ( !\nid{p}\, \mathtt{retry}.\, \tvar)}{\emptyset}{0}
	\end{array}\]
	System $\R_{TT}$ is not resilient with respect to $\Delta$ from Example \ref{ex:resilience} because the nodes add some communications to acknowledge correct interaction or coordinate on a retry iteration. It is however $n$-recoverable with $n = 6$.
\end{exa}

By Definition~\ref{recoveranew}, checking resilience and $n$-recoverability is reduced to the problem of checking weak barbed bisimulation.  Note that, in Definition~\ref{recoveranew}, the number of $\R'$ that can be reached from $\R$ is finite, because the execution up to $\R'$ lasts for $n$ time units and, by Proposition~\ref{nonzeno}, a system can perform only a finite number of actions in a finite amount of time.

In the following, we show that resilience is equivalent to $0$-recoverability. This fact was conjectured for not formally proven in~\cite{BLTV22}. This result is given in Section~\ref{sec:res0rec}.

\subsection{Resilience is equivalent to $0$-recoverability}
\label{sec:res0rec}

Equivalence of resilience and $0$-recoverability (Theorem~\ref{res0rec}) is based on two facts: 
\begin{enumerate}
\item for all initial systems $\R$, $(\R,\Delta)\gtrsim(\R,{\uparrow})$, given in Lemma~\ref{lem:0recoverability2}, and 
\item for all initial and $0$-recoverable systems $\R$, $(\R,\Delta)\lesssim(\R,{\uparrow})$, given directly in Theorem~\ref{res0rec}.
\end{enumerate}
Lemma~\ref{lem:0recoverability2} is based on a property that we call $\uparrow$-consistency. $\uparrow$-consistency correlates the syntactic structure cursed system and their corresponding counterparts with curse $\csup$, as they evolve.

\begin{defi}[$\uparrow$-consistency]
\label{ciao} 
Two systems $\R_{\Delta}$ and $\R_{\uparrow}$ are \emph{$\uparrow$-consistent} if there exist $\R$, $\R_u$, $\R_d$, 
and $\R_f$ such that 
\[\R_{\uparrow} = \R \pll \R_u \pll \R_f \quad  \quad \R_{\Delta} = \R \pll \R_d\] 
and: 
\begin{itemize}
\item $\R_u$ and $\R_d$ are parallel compositions of the same (possibly empty) set of nodes. 
\item the nodes in $\R_d$ are all down, i.e., of the form $\nodecp{n}{\csdown}{\emptyset}{\tval}{Q}$, 
\item $\R_f$ is the parallel composition of a (possibly empty) set of latent or floating messages.
\end{itemize}
\end{defi}

Intuitively, $\uparrow$-consistency defines a structural relation between the evolution of a system with and without curses: $\R$ models the parts of the system (if any) that $\R_{\uparrow}$ and $\R_{\Delta}$ have in common; 
$\R_u$ and $\R_d$ are the nodes that are up in $\R_{\uparrow}$ and down in $\R_{\Delta}$, respectively (they model the difference between $R_\uparrow$ and $R_\Delta$ wrt.\ crashed nodes); 
 moreover, $\R_{\uparrow}$ may have some additional floating messages, represented by $\R_f $, that have been lost in 
$\R_{\Delta}$.

$\uparrow$-consistency enjoys two properties. The first one, given in Lemma~\ref{ucon1}, is that 
instantaneous actions preserve $\uparrow$-consistency and does not decrease the number of down nodes in the cursed system. The second one, given in Lemma~\ref{ucon2}, ensures that the barbs of the cursed system are always a subset of those of the uncursed counterpart.

\begin{lem}\label{ucon1}
If $\R_{\Delta}$ and $\R_{\uparrow}$ are \emph{$\uparrow$-consistent} and 
$(\R_{\Delta}, \Delta) \actarrow{}(\R'_{\Delta}, \Delta)$ then
\begin{enumerate}
\item  there exists $\R'_{\csup}$ such that 
$(\R_{\csup}, \uparrow) \actarrow{}(\R'_{\csup}, \uparrow)$ and $\R'_{\Delta}$ and $\R'_{\uparrow}$ are \emph{$\uparrow$-consistent} 
\item the set of down nodes in $\R_{\Delta}$ is a subset of the set of down nodes in $\R_{\uparrow}$.
\end{enumerate}
\end{lem}
\begin{proof}
By induction on the derivation. In case of actions by [Snd], [Sched], [Rcv], and [Checkpoint] all yield that there is $\R'$ such that $\R'_{\csup}=\R' \pll \R_u \pll \R_f$ and $\R'_{\Delta}=\R' \pll \R_d$. Moreover, no down nodes are introduced in $\R'_{\Delta}$ and hence $\R'_{\Delta}$ and $\R'_{\uparrow}$ are \emph{$\uparrow$-consistent}. 
Rule [NUp] cannot be applied at time $0$. The only possible failure actions are [MsgDown] and [NDown]. 
In case for [MsgDown], there exist $\R'$ and $\R_f'$ such that $\R'_{\uparrow}=\R' \pll \R_u \pll \R_f'$ and $\R'_{\Delta}=\R' \pll \R_d$ where $\R'$ is as $\R$ but without the lost message, and $\R_f'$ is as $\R_f$ but with the addition of the lost message. The set of down nodes does not change, hence $\R'_{\Delta}$ and $\R'_{\uparrow}$ are {$\uparrow$-consistent}, yielding the thesis for this case. 
In case of [NDown], $\R'_{\uparrow}$ is of the form $\R' \pll \R'_u \pll \R_f$ and $\R'_{\uparrow}$ is for the form $\R' \pll \R'_d$ where $\R'$ is as $\R$ but without the node that went down, $\R'_d$ is as $\R_d$ but with the addition in parallel with the node that went down, similarly for $\R'_u$ but the node in this case is still up. The set of down nodes has increased in $\R'_{\Delta}$. It follows that $\R'_{\Delta}$ and $\R'_{\uparrow}$ are \emph{$\uparrow$-consistent}.

The cases for [Rec], [Str], and [ParCom] are immediate by induction. 
\end{proof}
 
\begin{lem}[$\uparrow$-consistency]\label{ucon2}
If $\R_{\Delta}$ and $\R_{\uparrow}$ are {$\uparrow$-consistent} then 
$\R_{\Delta}\mybarb{\brb}$ implies $\R_{\uparrow}\mybarb{\brb}$.
\end{lem}
 \begin{proof}
By $\uparrow$-consistency, there exist $\R$, $\R_u$, $\R_d$, and $\R_f$ such that $\R_{\uparrow} = \R \pll \R_u \pll \R_f$ and $\R_{\Delta} = \R \pll \R_d$. By Definition~\ref{barb} the barbs of 
$\R_{\uparrow}$ is the union of barbs of $\R$, $\R_u$, and $\R_f$, and the set of barbs of $\R_{\Delta}$ is the union of the barbs of $\R$ and $\R_d$. We only need to show that $\R_d$ does not have barbs that $\R_{\uparrow} $ does not have. This follows trivially from Definition~\ref{ciao} since $\R_d$ is the parallel composition of down nodes and hence has no barbs. 
\end{proof}

We can now prove a more general property of cursed systems at time $0$: 
\begin{lem}\label{lem:0recoverability2}
Let $\lesssim_{0}$ be the restriction of $\lesssim$ obtained considering only communication actions $\actarrow{}$ in Figure~\ref{fig:reduction} (i.e., no time-consuming actions $\timearrow{}$) on systems $\R$ such that $\stime(\R)=0$. It holds that 
\[(\R, \uparrow)\lesssim_{0}(\R,\Delta)\]
\end{lem}
\begin{proof}
By coinduction, observing that initial systems are {$\uparrow$-consistent}, {$\uparrow$-consistency} is preserved by communication actions by Lemma~\ref{ucon1} and ensures that $\R_{\Delta}\mybarb{\brb}$ implies $\R_{\uparrow}\mybarb{\brb}$ by Lemma~\ref{ucon2}.
\end{proof}

We next show an intuitive property that will be useful to show equivalence of resilience and $0$-reliability: $(\R,\Delta)$ and $(\R,\csup)$ are weak barbed bisimilar at time $0$ if no time actions or failures happen (Lemma~\ref{lem:0recoverability}). This is proved by coinduction via Lemma~\ref{lem:0recoverability0} ensuring that actor/node transitions preserve equivalence of barbs in the evolution of cursed and uncursed systems.

\begin{lem}\label{lem:0recoverability0}
If $\stime(\R)=0$, $\R$ is fail-free, and $\actarrow{}$ is an actor/node action then 
\begin{enumerate}
\item $(\R,\Delta)\actarrow{}(\R',\Delta)\Leftrightarrow(\R,\csup)\actarrow{}(\R',\csup)$
\item $(\R,\Delta)\actarrow{}(\R',\Delta) \Rightarrow$ $\R'$ is fail-free. 
\end{enumerate}
\end{lem}
\begin{proof}
By induction on the derivation proceeding by case analysis on the last rule used. The base cases, for rules [Snd], [Sched], [Rcv], and [Checkpoint], are mechanical. The inductive cases for rules [Rec], [Str], and [ParCom], are straightforward by inductive hypothesis. 
\end{proof}

\begin{lem}\label{lem:0recoverability}
Let $\approx_{0}^{\uparrow}$ be the restriction of $\approx$ obtained considering only actor/node actions in Figure~\ref{fig:reduction} (i.e., no failure and no time-consuming actions) and $\stime(\R)=0$ with $\R$ fail-free. It holds that 
\[(\R, \Delta)\approx_{0}^{\uparrow}(\R,\uparrow)\]
\end{lem}
\begin{proof}
This lemma holds by coinduction, observing that $\stime(\R)=0$ and hence, by Lemma~\ref{lem:0recoverability0}, $(\R,\Delta)\actarrow{}(\R',\Delta)$ if and only if $(\R,\csup)\actarrow{}(\R',\csup)$ with $\R'$ fail-free. 
\end{proof}

We are not able to state the main results: equivalence of resilience and $0$-recoverability. 

\begin{thm}[$0$-recoverability and resilience]\label{res0rec}
An initial cursed system $(\R,\Delta)$ is resilient if and only if it is $0$-recoverable. 
\end{thm}
\begin{proof}
The only if case is immediate since resilience implies the existence of a state, the initial one, such that $(\R,\uparrow) \approx (\R,\Delta)$. For the if case, we assume $(\R,\Delta)$ to be \emph{$0$-recoverable}: for all path of executions of $(\R,\Delta)$ at time $0$ (i.e., $0$-paths of $(\R,\Delta)$)  there exists a state $(\R',\Delta)$ in that path such that $time(\R')=0$ and $(\R,\uparrow) \approx (\R',\Delta)$. Fix a path $(\R,\Delta)\reducesStar{}(\R',\Delta)$.
Consider a generic intermediate state $(\R'',\Delta)$ such that $(\R,\Delta)\reducesStar{}  (\R'',\Delta)\reducesStar{} (\R',\Delta)$. The reductions to $(\R'',\Delta)$ can be by either (i) one of the Actor/node actions, or by (ii) one of the instantaneous failure actions ([MsgLoss] or [NodeDown]). 

Observe that, for all $\R$ and $\Delta$, the relation $((\R,\Delta), (\R, \uparrow))$ is a bisimulation if we consider a restriction of the reduction relation that only uses actions that is actor/node actions with no failure and no time-consuming actions (Lemma~\ref{lem:0recoverability}, using the fact that $\R$ is initial and hence fail-free). So, if only actions (i) are possible in the reduction of $(\R,\Delta)$ then $(\R'',\Delta) \approx (\R'',\uparrow)$ by Lemma~\ref{lem:0recoverability}. If there are only (i) actions at time $0$, since any state reachable from $(\R,\Delta)$ at time $0$ is bisimilar to the corresponding state reached by $(\R,\uparrow)$ then $(\R, \Delta)$ is resilient. Hence done. 

The argument proceeds similarly in case of (ii) actions that preserve the corresponding barbs of the system with $\Delta$ and the one with $\csup$. Assume therefore that, possibly after a number of barb-preserving reductions by (i) or (ii), the system with $\Delta$ and the one with $\csup$ reach states $(\R'',\Delta)$ and $(\R''',\csup)$, respectively, where $\R''$ and $\R'''$ have different barbs. By Lemma~\ref{ucon2} it can only be the case $\R''\not\mybarb{\brb}$ and $\R''' \mybarb{\brb}$. By hypothesis we have that $(\R'',\Delta)\actarrow{}^\ast(\R',\Delta)\approx(\R,\uparrow)$. We also know that $(\R,\uparrow)\actarrow{}^\ast(\R''',\csup)$ with $\R'''\mybarb{\brb}$. Hence $(\R',\Delta)$ can reach a state that has barb $\brb$ and is bisimilar to $(\R''',\csup)$. This shows $(\R,\Delta)\gtrsim(\R,\csup)$. The fact that $(\R,\Delta)\lesssim(\R,\csup)$ by Lemma~\ref{lem:0recoverability2} yields the thesis $(\R,\Delta)\approx(\R,\csup)$. 
\end{proof}

\subsection{Recoverability for checkpointing systems}
\label{sec:recoverability-cp}

The definition of recoverability in the previous section formalises a system restarting from the initial state, and does not capture checkpointing systems that recover to intermediate states. In this section we add definition of bisimulation up to a particular time, and also a notion of $n$-recoverability for checkpointing systems. This is illustrated with an example of a system that is not $n$-recoverable but that is $n$-checkpoint recoverable.

We introduce a notion of weak barbed simulation up to $n$ where $n$ is a relative time, up to which we want to compare behaviour (ignoring what happens afterwards). 

\medskip
\noindent
\begin{defi}[Weak barbed simulation up to $n$]\label{sim-to-n}
Recall $\reduces{}\in \{\actarrow{}, \timearrow\}$.
A weak barbed simulation up to $n$ is a set of binary relations $\simu^r$ for $r\leq n$  between cursed systems such that: 
\begin{enumerate}
\item $r=0$: $(\R_1,\Delta_1) \simu^0 (\R_2,\Delta_2)$ for all $\R_1$, $\Delta_1$, $\R_2$ and $\Delta_2$;
\item $r>0$ and $(\R_1,\Delta_1) \simu^r (\R_2,\Delta_2)$ implies:
\begin{enumerate}
\item If $(\R_1,\Delta_1)\reduces{} (\R_1',\Delta_1)$ and $s = time(\R_1') - time(\R_1) \leq r$, then there exists $\R_2'$ such that $(\R_2,\Delta_2) \reducesStar{}(\R_2',\Delta_2)$ and $(\R_1',\Delta_1) \simu^{(r-s)} (\R_2',\Delta_2)$.
\item If $\R_1 \mybarb{\brb}$ for some $\brb$, then there exists $\R_2'$ such that $(\R_2,\Delta_2) \reducesStar{} (\R_2',\Delta_2)$ and $\R_2'\mybarb{\brb}$.
\end{enumerate}
\end{enumerate}
\noindent We say $(\R_1,\Delta_1)$ is weak barbed similar to $(\R_2,\Delta_2)$ up to $n$, written $(\R_1,\Delta_1) \lesssim^n (\R_2,\Delta_2)$, if there exists some weak barbed simulation up to $n$, $\simu^n$, such that $(\R_1,\Delta_1) \simu^n (\R_2,\Delta_2)$. By point (1), weak barbed simulation up to $0$ holds for all pairs of systems, whereas weak barbed simulation is morally equivalent to weak barbed simulation up to $\infty$.
 \end{defi}

\begin{defi}[Weak barbed bisimulation up to $n$]\label{bisim-to-n}
We say that $\simu^n$ is a weak barbed bisimulation up to $n$ if $\simu^n$ and ${\simu^n}^{-1}$ are weak barbed simulations up to $n$. We say that $(\R_1,\Delta_1)$ and $(\R_2,\Delta_2)$ are weak barbed bisimilar up to $n$, written  $(\R_1,\Delta_1) \approx^n (\R_2,\Delta_2)$, if there exists some weak barbed bisimulation up to $n$, $\simu^n$, such that $(\R_1,\Delta_1) \simu^n (\R_2,\Delta_2)$.
\end{defi}
 
It is a straightforward consequence of these definitions that if two systems are (bi-)similar up to $n$ then they are (bi-)similar up to $r$ for any $r<n$, and two systems are (bi-)similar  if and only if they are (bi-)similar up to $n$ for all $n$.

We now define a variant of $n$-recoverability that, based on weak barbed bisimulation up to $n$, aims to characterise recoverability for systems that use checkpoints to recover from failures. 

\begin{defi}[$n$-checkpoint-recoverable]\label{n-rec-plus}
For all $\R''$ reachable from $(\R,\csup)$ such that $\stime(\R'')<n$ and $(\R,\csup) \approx^{\stime(\R'')}(\R,\Delta)$ there exists  $\R'$ such that $\stime(\R')<n$ and $(\R,\Delta)\reducesStar{}(\R',\Delta)$ and $(\R',\Delta)\approx(\R'',\csup)$.
\end{defi}

Informally, a cursed system $(\R,\Delta)$ that behaves correctly up to a state at time $t$, can always reach a later state $(\R',\Delta)$ which is bisimilar to the state $(\R'',\csup)$. Suppose that $(\R,\csup) \approx^t(\R,\Delta)$, where $t\leq n$. Then for any $\R''$ reachable from $(\R,\csup)$ such that $\stime(\R'')=t$ there exists $\R'$ that by time $n$ displays the remaining behaviour of the correct (uncursed) system, that is $(\R'',\csup)$. In contrast to Definition~\ref{recoveranew} this definition does not require that $(\R',\Delta)$ exhibits the complete behaviour of $(\R,\csup)$ but only the behaviour after a certain point, for example from a checkpoint onwards. In contrast to Definition~\ref{recoveranew}, Definition~\ref{n-rec-plus} does not require the recovered system $(\R' ,\Delta)$ to exhibit the complete behaviour of the uncursed system since its initial state (i.e., have the same behaviour of $(\R,\uparrow)$) but only the behaviour of $(\R,\uparrow )$ after a certain point (i.e., $(\R'',\uparrow)$) for example from a checkpoint onwards.

\begin{exa}[$n$-checkpoint-recoverability]\label{ex:checkpoint}
	Consider a variant $\R_{TT}$ of $\R_T$ from Example~\ref{ex:redundance}, latency of $1$ time unit, and $\Delta$ that curses node \(\mathtt{p}\) to go down at time 2: 
\[\begin{array}{lll}
\R_{TT} =   & \node{p}{?\mathtt{order}.\savek.\sleep{}.!\mathtt{c}\,\mathtt{item}.\, ? \mathtt{ok}\, . 0 \after\, \color{red}3\color{black} ~  0 }{\emptyset}{0}\parallel
\\
&\node{c}{!\mathtt{p}\,\mathtt{order}.?\mathtt{item}.\ !\mathtt{p}\,\mathtt{ok}. 0 \after\, \color{red}3\color{black} ~  (!\mathtt{p}\,\mathtt{failed}.?\mathtt{item}.\ !\mathtt{p}\,\mathtt{ok}. 0 \after\, \color{red}5\color{black} ~  0)}{\emptyset}{0}
\end{array}
\]

	System $\R_{TT}$ used checkpointing to restart from an intermediate state in the event of failure. In the case of the node \(\mathtt{p}\) to going down at time 2, the system reduces in a number of steps to:
	\[\begin{array}{lll}
		\R'_{TT} = & \node{p}{\sleep{}.!\mathtt{c}\,\mathtt{item}.\, ? \mathtt{ok}\, . 0 \after\, \color{red}3\color{black} ~  0 }{\emptyset}{3}\parallel
		\\
		&\node{c}{?\mathtt{item}.\ !\mathtt{p}\,\mathtt{ok}. 0 \after\, \color{red}5\color{black} ~  0)}{\emptyset}{3} \parallel 1.\PM{c}{p}{\mathtt {failed}}{3}
	\end{array}
	\]
	It is this state that makes system $\R_{TT}$ not $n$-recoverable with respect to $\Delta$ both because node \(\mathtt{c}\) sends message $\mathtt{failed}$ in its timeout process and the communication of the \(\mathtt{order}\) message does not get repeated when the rest of interaction is repeated. While the $\mathtt{failed}$ message is not read by \(\mathtt{p}\), it adds an additional barb to the cursed system $(\R_{TT},\Delta)$ that is not matched by the uncursed system $(\R_{TT},\csup)$. 
			
	The system is however $n$-checkpoint-recoverable with $n = 4$, $\R'_{TT}$ further reduces to 
	\[\begin{array}{lll}
		& \node{p}{!\mathtt{c}\,\mathtt{item}.\, ? \mathtt{ok}\, . 0 \after\, \color{red}2\color{black} ~  0 }{\mathtt{failed}}{4}\parallel \node{c}{?\mathtt{item}.\ !\mathtt{p}\,\mathtt{ok}. 0 \after\, \color{red}4\color{black} ~  0)}{\emptyset}{4}
	\end{array}
	\]
	from which state the cursed system exhibits the behaviour of the uncursed system from time 2 (or from the checkpoint) onwards.
\end{exa}

\subsection{Fault-tolerance: a more general perspective}
\label{sec:fault}

In \cite{10.1145/311531.311532}, the author gives a theoretical definition of the problem of fault tolerance along two dimensions: safety (the system does not reach bad states, although it can possibly stop due to faults) and liveness (the system eventually reaches good states, hence in case of bad behaviour it eventually recovers). In this context, the guarantee of both safety and liveness is called \emph{masking fault-tolerance}, of only safety is called \emph{fail-safe fault-tolerance}, and of only liveness is called \emph{non-masking fault-tolerance}.

In our framework, we can characterise these three kinds of fault-tolerance by using our simulation relation: 
\begin{itemize}
\item $(\R,\Delta)\approx (\R,\uparrow)$ - \emph{masking fault-tolerance}: $\R$ cursed by $\Delta$ has all and only the behaviour of healthy system $\R$. 
\item $(\R,\Delta)\lesssim (\R,\uparrow)$ - \emph{fail-safe}: $\R$ cursed by $\Delta$ has only the behaviour of healthy system $\R$. 
\item $(\R,\Delta)\gtrsim (\R,\uparrow)$ - \emph{non-masking fault tolerant}: $\R$ cursed by $\Delta$ has only the behaviour of healthy system $\R$. 
\end{itemize}

Resilience, given in Definition~\ref{def:resilience}, corresponds to the safety and liveness combination of masking fault-tolerance. We have shown in Example~\ref{ex:redundance} that masking-failure can be attained by using space redundancy (e.g., replication of nodes as in the multiple producers scenario) and time redundancy (e.g., retry-strategies). In fact, the author in \cite{10.1145/311531.311532} substantiates that redundancy is a necessary condition for fault tolerance. 
Fail-safe fault tolerance, while easier to attain, is not the most desirable property in many real-world scenarios: a systems that just stops to prevent `bad' actions, may not be a suitable model when you want eventually consistency despite perturbations to the ideal course of actions. 
Intuitively, both fail-safe fault tolerance and non-masking fault tolerance for cursed system can be expressed by using the notion of weak barbed simulation given in Definition~\ref{simulation}, as shown above.  
The formulation of \emph{non-masking fault tolerant} as $(\R,\Delta)\gtrsim (\R,\uparrow)$ is very general. In principle, this definition consider fault tolerant any system that performs an infinite sequence of actions among which, sometimes, a correct action happens to make the system progress. Practically we would want to see that the behaviour added is not random, but follows a sensible pattern of restart, reset or another benign behaviour. In this paper we have put most emphasis on non-masking fault tolerance, but focussing on more stringent definitions of non-masking fault-tolerance: some unforeseen sequence of actions may be visible at some point, but after some recovery actions, at a time that is not later than $n$, the system will revert to the required behaviour by restarting from the beginning ($n$-recoverability for reset systems) or from the point of failure ($n$-recoverability for checkpointing systems). These definitions are intentionally non-general, with the aim of capturing known recovery patterns. We leave as a future work the extension of $n$-recoverability to cater for periodic failures. 

A similar approach, of characterising masking/fail-safe/non-masking fault tolerance using simulation was followed by 
 \cite{10.1007/s00165-017-0426-2} but with a clear distinction of good versus faulty states (using coloured Kripke structures). More on the relationship with  \cite{10.1007/s00165-017-0426-2} is discussed in Section \ref{sec:related}.

\section{Augmentation of cursed systems}
\label{sec:augmentation}

Augmentation of a cursed system is the result of adding or modifying some behaviour in the initial system to improve the system's ability of handling failures. The following definition applies to reset systems as it is based on $n$-recoverability. A corresponding notion of augmentation could be given for checkpointing systems by using, in Definition~\ref{def:augmentation}, $n$-checkpoint-recoverability instead of $n$-recoverability. In the remaining of this section we focus on reset systems.

\begin{defi}[Augmentation]\label{def:augmentation}
$\R_{\mathtt I}$ is an augmentation of $\R$ if $time(\R_{\mathtt I}) = time(\R)$ and:
\begin{description}
\item[i) transparency] 	$(\R,\uparrow) \approx (\R_{\mathtt I},\uparrow)$
\item[ii) improvement] there exist $\Delta$ and $n$ such that $(\R_{\mathtt I},\Delta)$ is $n$-recoverable and $(\R,\Delta)$ is not $n$-recoverable.
\end{description}
Moreover, we say that an augmentation is \emph{preserving} if, for all $n$ and $\Delta$,
$(\R,\Delta)$ is $n$-recoverable implies $(\R_{\mathtt I},\Delta)$ is $n$-recoverable.
\end{defi}

\begin{exa}[Augmentation]\label{ex:augmentation}
Consider the small producer-consumer system $\R$ below, composed of a producer node $\nid{n_p}$, a queue node $\nid{n_q}$, and a consumer node $\nid{n_c}$. The producer recursively sends items to the queue and sleeps for a time unit. The queue expects to receive an item within three time units that then gets sent to the consumer. In case of a timeout the queue loops back to the beginning and awaits an item from the producer. The consumer recursively receives items from the queue. We fix the latency of the system to $L = 1$.
\begin{align*}
\R = \, & \nid{n_q}[\rec{t} \rcv {item. \sleepk. \snd \nid{n_c} \, item.\rvar{t}}
		\after{3} \rvar{t} ]{\color{blue}(\emptyset)}{\color{purple}(0)} \pll
		\\
     &\node{n_{p}}{\rec{t} \snd \nid{n_q} \, item. \sleepk.\rvar{t}}{\emptyset}{0}
	  \pll
	  \node{n_c}{\rec{t} \rcv {item.\sleepk.\rvar{t}} \after{4} \rvar{t}}{\emptyset}{0}
\\[0.2cm]
\R_{\mathtt I} = \, & \R \pll
\node{n_{p'}}{\rec{t} \snd \nid{n_q} \, item. \sleepk.\rvar{t}}{\emptyset}{0}
\end{align*}
The augmented producer-consumer $\R_{\mathtt I}$ adds behaviour to the system by having a second producer node $\nid{n_{p'}}$.
$\R_{\mathtt I}$ improves the resilience to a producer node or its link failing or being slow. For example the curse function $\Delta(\nid{n_p})$ injecting node delay for the producer node between time 1 and 3 and $\csup$ otherwise impacts the first system $\R$ but not its augmented counterpart $\R_{\mathtt I}$. $\R$ is $4$-recoverable while $\R_{\mathtt I}$ is $0$-recoverable. Moreover, $\R_{\mathtt I}$ preserving augmentation of system $\R$.
\end{exa}

\subsection{Augmentation with scoped barbs}
\label{scopes}
Augmentations often need to introduce additional behaviour into actors. One may want to disregard part of `behind the scenes' augmentation when comparing the behaviour of cursed systems using the relation in Definition~\ref{bisimulation}.
For simplicity, instead of adding scope restriction to the calculus, we extend barbs with scopes to hide behaviour of some nodes or links.
With mailboxes, all interactions to a node are directed to the one mailbox. Defining scope restriction only on node identifiers would be less expressive than scope restriction based on channels, e.g., it would not be possible to hide specific communications to a node, while in channel-based calculi one can use ad-hoc hidden channels.
To retain expressiveness, we define scope restriction that takes into account \emph{patterns} in the communication
between nodes.

\begin{defi}[Scoped barb]\label{def:scoped_barb}
Let $N$ be a finite set of elements of the form $!\,  \nid n \, p$ or $?\,  \nid n \, p$ where $\nid n\in \nodes$ and $p$ is a pattern. $\R\sbarb \brb $ if: (1) $\R\mybarb \brb $, (2) $\brb \not \in N$, and (3) if $\brb = \snd \nid n\,  m$ then for all $\snd \nid n \,  p\in N$,  $\nmatch{p}{m}$.
If $\R \sbarb{\brb}$ we say that $\R$ has a $N$-scoped \emph{barb} on $\brb$.
\end{defi}

We extend Def.~\ref{bisimulation} using $\sbarb{}$ instead of $\mybarb{}$, obtaining scoped weak-barbed bisimulation $\approx_N$, and Def.~\ref{def:augmentation} to use $\approx_N$.
This setting allow us to analyse producer consumer scenarios, or more complex ones, like the Circuit Breaker pattern~\cite{nygard2018release} widely used in distributed systems.

\begin{exa}[Circuit breaker]\label{cb} Consider system \((\R, \Delta)\) with a client $\nid{n_c}$ and a service $\nid{n_s}$, and its augmentation $\R_{\mathtt I}$ with a circuit breaker running on node $\nid{n_{s}}$:
\[
\begin{array}{ll}
	\R = & \node{n_{c}}{\rec{t} \snd {\nid{n_s} \, request}. \rcv {reply. \sleepk.\rvar{t}} \after{4} \inact}{\emptyset}{0}
	\pll
	\\
	& \node{n_s}{\rec{t} \rcv request. \sleepk. \snd \nid{n_c} \, reply.\rvar{t} \after{4} \rvar{t}}{\emptyset}{0}
	\\[0.3cm]
\R_{\mathtt I} = &
	\node{n_{c}}{\rec{t} \snd {\nid{n_{s}} \, request}.
		\rcv
		\setof{reply. \sleepk.\rvar{t}, \, ko. P_f}
		\after{8} \inact
	}{\emptyset}{0} \pll
	\\
	&  \node{n_{s}}{ \rec{t} \rcv {X_1}. ! \nid{n_1} \, X_1.  \rcv X_2. ! {\nid{n_c}\, X_2}. \rvar{t}
		\after{4} P'_f \after{4} \rvar{t}}{\emptyset}{0} \pll
	\\
	\multicolumn{2}{l}{ \quad \node{n_1}{
			\rec{t} \rcv
			\setof{
				request. \sleepk. ! \nid{n_{s}} \, reply.\rvar{t}, \,
				ruok. \sleepk. ! \nid{n_{s}} \, imok. \rvar{t}
			}
			\after{6} \rvar{t}
		}{\emptyset}{0}
	} \\[0.2cm]
	P_f = & \rec{t'} \rcv retry. \sleepk. \rvar{t} \after{5} \rvar{t'} \\
	P'_f = & \snd {\nid{n_c}\, ko}.\sleepk. \rec{t'} \snd {\nid{n_s}\, ruok}.
		\rcv {imok. \sleepk. \snd {\nid{n_c}\, retry}.\rvar{t}} \after{3} \rvar{t'}
\end{array}
\]

with a $\Delta(\nid{n_c}, \nid{n_s})$ injecting link slow $\csslow$ at times $1$, $2$, and $3$ and healthy otherwise, and latency to $L = 1$. The impact of failure on the $\R$ makes it unrecoverable, as the link delay cascades to node $\nid{n_c}$.
We augment $\R$ with a circuit breaker process which runs on the previous server node $\nid{n_{s}}$ that monitors for failure, prevents faults in one part of the system and controls the retries to the service node now $\nid{n_{1}}$. The node $\nid{n_{s}}$ forwards messages between nodes $\nid{n_{c}}$ and $\nid{n_{1}}$, and in case of a timeout checks the health of $\nid{n_{s}}$ and tells node $\nid{n_c}$ when it can safely retry the request. When comparing $\R$ and $\R_{\mathtt I}$ for resilience, recoverability or transparency we wish to abstract from the additional behaviour introduced by the circuit breaker pattern for which we use Def.~\ref{def:scoped_barb} with: \(N = \{\snd \nid{n_s} \, ruok, \, \rcv \nid{n_s} \, imok, \, \rcv \nid{n_s} \, reply, \, \rcv \nid{n_1} \, request, \, \snd \nid{n_s} \, reply, \, \rcv \nid{n_1} \, ruok, \, \snd{n_s} \, imok, \) \\
	\(  \snd{n_c} \, ko, \,\snd{n_c} \, retry, \, \rcv{n_c} \, ko, \, \rcv{n_c} \, retry \}\).
This effectively hides the entire behaviour of $\nid{n_1}$ and node $\nid{n_s}$'s health checking behaviour. Using the extended definition we find that for the same curse function system $\R_{\mathtt I}$ is $0$-recoverable. Similarly, for
the curse function delays link $(\nid{n_s}, \nid{n_1}$) at times $1$, $2$, and $3$, $\R_{\mathtt I}$ is $0$-recoverable.
\end{exa}

 \section{Prospective applications}
\label{sec:app}

This work is a first step towards an analysis of mailbox systems with failures and has the purpose of clarifying the problem space. An informal validation of the relevance of the work was attained through interaction with our industry partners, in particular Erlang Solutions Ltd. and Actyx AG, as well as in applying it to a collection of real-world case studies and patterns, such as the circuit breaker in Section \ref{cb}.  

To support analysis and development of real-world systems, we aim to build on the current work. In this section we discuss two potential applications; their development goes beyond the scope of the formal setting given in the current work.

\subsection{Analysis of cursed systems} Encoding the models of failures and systems into verification tools like UPPAAL is fairly straightforward. We provide a prototype encoding here. A straightforward encoding only supports analysis of a system against a specific $\Delta$ `traces', and so, by repeating this analysis, to a limited set of curse cases. 

As a more powerful development, we are working on generalizing the notion of $\Delta$ to a symbolic entity, that can finitely characterise infinite patterns, together with a tractable algorithm to determine simulation that is parametric with respect to this symbolic $\Delta$. This feature would allow it to be determined whether a system model is resilient with respect to a given set of curses, or synthesise the curses that a system can or cannot deal with. 

Code generation or synthesis would, in turn support top-down or bottom-up development, (respectively). Existing approaches to code generation provide seamless links between process-calculi-based models and Erlang code. For example, the tool described in~\cite{DBLP:journals/programming/BocchiOV23}, which presents a proof of concept of a theoretical advance, can generate Erlang $\mathtt{gen\_statem}$ code from a process-calculus specification and extract specifications from Erlang $\mathtt{gen\_statem}$ code. The circular transformation described above is possible thanks to the code structure induced by Erlang $\mathtt{gen\_statem}$ itself, which yields modular code that is structured as a finite state machine and hence has a  straightforward correspondence with its model. A similar approach could be taken to our modelling of failure scenarios by supporting a richer process calculus that includes time and  timeouts.

\subsection{Test support}
A second direction is to use property-based testing (PBT), as implemented by QuickCheck~\cite{quickCheck}, initially for Haskell and Erlang, and subsequently for a variety of other languages. Property-based testing replaces unit tests by tests of logical properties of the system under test (SUT), expressed in a universal fragment of first-order logic. A universal property is evaluated at a randomly generated set of values, and any counter-example is systematically shrunk to a simplest such example, according to some size metric. Successful application of property-based testing therefore depends on three things: being able to express relevant properties of a system in a logical form; being able to generate values from relevant domains in a way that optimises coverage of the domain; and being able to ``shrink'' values in an effective and efficient way. PBT can be seen as a complement to more heavyweight verification approaches: for example, it is worthwhile subjecting a candidate theorem to PBT before embarking on developing a formal proof.

Stateful systems in Erlang~\cite{QuickCheckPulse} and other languages can subject to PBT using state machine models. The state machine provides an abstract model of the system, and is used to guide testing of the SUT: random sequences of transitions of the state machine exercise the SUT, and shrinking simplifies and shortens counter-example traces. 

In the context of the work presented here, QuickCheck can be used to test systems in which failure is modelled explicitly in a state machine model, but could also be extended to include  modelling of the symbolic $\Delta$ function discussed above -- e.g.\ using logical constraints -- and to generate and shrink instances of $\Delta$ with particular properties.

 \section{Conclusion and related work}
\label{sec:related}

We introduced a model for actor-based systems with grey failures and investigated the definition of behavioural equivalence for it. We used weak barbed bisimulation to compare systems on the basis of their ability to recover from faults, and defined properties of resilience, recoverability and augmentation.  We reduced the problem of checking reliability properties of systems to a problem of checking bisimulation. We introduced scope restriction for mailboxes based on patterns, which allows us to model relatively complex real-world scenarios like the Circuit Breaker. 

As further work we plan to extend the recovery function $\pf$ to model check-pointing of intermediate  node states. Note that $\pf$ can already be set as an arbitrary process, but a more meaningful extension would account for the way in which checkpoints are saved. Moreover, we plan to add a notion of intermittent correctness, to model recovery with partial checkpoints rather than re-starting from the initial state, or intermittent expected/unexpected behaviour.
Another area of future work is to use the characteristic formulae approach~\cite{GrafS86a,Steffen89}, a method to compute simulation-like relations in process algebras, to generate formulae for the properties introduced and reduce them to a model checking problem that can be offloaded to a model checker. 

A related formalism to our model is Timed Rebeca~\cite{Aceto_2011}, which is actor-based and features similar constructs for deadlines and delays. Timed Rebeca actors can also use a `\emph{now}' function to get their local times. Extending our calculus with `\emph{now}' and allowing messages to have time as data sort, would allow us to model scenarios e.g., where a node calculates the return-trip time to another node and changes its behaviour accordingly. While Timed Rebeca can encode network delays (adding delays to receive actions -- using a construct called `\emph{after}'), it does not model links explicitly. Explicit links and separation between curses and systems make it easier in our calculus to compare systems with respect to recoverability. Rebeca was encoded in McErlang~\cite{Aceto_2011} and Real-Time Maude~\cite{SABAHIKAVIANI201585} for verification. We have ongoing work on encoding our model in UPPAAL. Our main challenge in this respect is to formalise a meaningful and manageable set of curses to verify the model against.

In~\cite{DBLP:journals/jlp/FrancalanzaH07}, Francalanza and Hennessy introduced a behavioural theory for D$\pi$F, a distributed $\pi$-calculus with with nodes and links failures. For a subset of D$\pi$F, they also developed a notion of fault-tolerance up to $n$-faults~\cite{DBLP:journals/iandc/FrancalanzaH08}, which is preserved by contexts, and which is related to our notion of resilience. 
The behavioural theory in~\cite{DBLP:journals/jlp/FrancalanzaH07} is based on reduction barbed congruence. The idea is to use a contextual relation to abstract from the behaviour of hidden nodes/links, while still observing their effects on the network, e.g., as to accessibility and reachability of other nodes. The scoped barbs in Section~\ref{scopes} have the similar purpose of hiding augmentations while observing their effects on recoverability. However, because of asynchronous communication over mailboxes (while D$\pi$F is based on synchronous message passing), our notion of hiding is less structural (i.e., based on nodes and links) and more application-dependent (i.e., based on patterns). At present, we have left pattern hiding out of the semantics, but further investigation towards a contextual relation that works for hidden patterns is promising future work. 
D$\pi$F studies partial failures but does not consider transient failures and time. On the other hand, D$\pi$F features mobility which we do not support.  
In fact, we rely on the assumption of fixed networks: since our observation is based on patterns (and ignores senders) we opted for relying on a stable structure to simplify our reasoning on what augmentation vs recoverability means, leaving mobility issues for future investigation.

Most ingredients of the given model (e.g., timeouts~\cite{DBLP:conf/fossacs/LaneveZ05,DBLP:conf/aplas/BergerY07,DBLP:conf/wsfm/LopezP11}, mailboxes \cite{MostrousV11}, localities~\cite{ICALP-1997-RielyH}\cite{BergerH00}\cite{Castellani01}) have been studied in literature, often in isolation. We investigated the inter-play of these ingredients, focussing on reliability properties. 
One of the first papers dealing with asynchronous communication in process algebra is by de Boer et
al.~\cite{BoerKP92}, where different observation criteria are studied (bisimulation, traces and
abstract traces) following the axiomatic approach typical of the process algebra
ACP~\cite{BergstraK84}. An alternative approach has been followed by Amadio et al.~\cite{AmadioCS98} who
defined asynchronous bisimulation for the $\pi$-calculus~\cite{MilnerPW92a}. They started from
operational semantics (expressed as a standard labelled transition system), and
then considered the largest bisimulation defined on internal steps that equates
processes only when they have the same observables, and which is closed under
contexts. The equivalence obtained in this way is called barbed congruence~\cite{MilnerS92}.
Notably, when asynchronous communication is considered, barbed congruence
is defined assuming as observables the messages that are ready to be delivered
to a potential external observer.
Merro and Sangiorgi~\cite{MerroS98} have subsequently studied barbed congruence in
the context of the Asynchronous Localised $\pi$-calculus (AL$\pi$), a fragment of the asynchronous $\pi$-calculus in which only output capabilities can be
transmitted, i.e., when a process receives the name of a channel, it can only send
messages along it, but cannot receive on it. Another line of research
deals with applying the testing approach to asynchronous communication; this has been investigated by Castellani and Hennessy~\cite{CastellaniH98} and by Boreale
et al.~\cite{BorealeNP99,BorealeNP02}. These papers consider an asynchronous variant of CCS~\cite{Milner89}. Testing discriminates less than our equivalence, concerning choice, and observes divergent behaviours which we abstract from. 
Lanese et al.~\cite{LaneseSZ19} look at bisimulation for Erlang, focussing on the management of process ids.
Besides the aforementioned work by Francalanza and Hennessy~\cite{DBLP:journals/jlp/FrancalanzaH07,DBLP:journals/iandc/FrancalanzaH08}, several works look at distributed process algebras with unreliable communication due to faults in the underlying network. Riely and Hennessy~\cite{ICALP-1997-RielyH} study behavioural equivalence over process calculi with locations.
Amadio~\cite{Amadio97} extends the $\pi$-calculus with located actions, in the context of a higher-order distributed programming language. Fournet et al.~\cite{FournetGLMR96} look at
locations, mobility and the possibility of location failure in the distributed join calculus. The failure of a location can be detected and recovered from. Berger and Honda~\cite{BergerH00} augment the asynchronous $\pi$-calculus with a timer, locations, message-loss, location failure and the ability to save process state. They define a notion of weak bisimulation over networks. Their model however does not include timeout, link delays, or a way of injecting faults. Cano et al.~\cite{cano2019multiparty} develop a calculus and type system for multiparty reactive systems that models time dependent interactions. Their setting is synchronous and their focus is on proving properties as types safety or input timeliness, while ours is comparing asynchronous systems with faults.

In Section~\ref{sec:fault} we discussed two related works: one characterising fault-tolerance using safety and liveness ~\cite{10.1145/311531.311532}, and one~\cite{10.1007/s00165-017-0426-2} instantiating such characterisation using simulation.
One of the main differences of our work with~\cite{10.1007/s00165-017-0426-2} is the communication model: the work in \cite{10.1007/s00165-017-0426-2} uses coloured Kripke structures, whereas we use an asynchronous process calculus with explicit actor-based features.
Another difference is the model of failure and the characterization of good versus bad states. Usually, fault-tolerance is studied against a well define model of failure and set of failure scenarios. Systems that can recover from arbitrary failures, called \emph{self-stabilizing} \cite{10.1145/311531.311532} are difficult to build and verify. Hence models for foult-tolerant systems normally have a model of failures within. In our case, this would be a set of $\Delta$ functions. In \cite{10.1007/s00165-017-0426-2}  there is an explicit and statical labelling (colouring) of each state as good or bad. In contrast, we observe deviations from the behaviour of the same system but without failures, so that a bad state is actually one that breaks the bisimulation relation with the corresponding uncursed system. 
Finally, in \cite{10.1145/311531.311532}\cite{10.1007/s00165-017-0426-2}, there is no clear distinction between the notion of fault -- a defect of the system -- and its concrete manifestation as a symptom. Our model of failure, that separates system from curses yields a more agnostic view of what a bad state is that supports modular reasoning on the relationship between causes and symptoms.

\bibliographystyle{alphaurl}
\bibliography{main}
\end{document}